\documentclass[12pt,a4paper]{iopart}
\pdfoutput=1
\usepackage{iopams,setstack,mathrsfs,amsthm,cite,enumerate,float,graphicx,slashed}
\usepackage[utf8]{inputenc}
\usepackage[T1]{fontenc}
\usepackage{lmodern}
\usepackage[colorlinks,linkcolor=blue,citecolor=blue,urlcolor=blue]{hyperref}
\newcommand{\al}{\alpha}
\newcommand{\be}{\beta}
\newcommand{\de}{\delta}

\newcommand{\vep}{\varepsilon}
\newcommand{\ga}{\gamma}
\newcommand{\ka}{\kappa}
\newcommand{\la}{\lambda}
\newcommand{\om}{\omega}
\newcommand{\si}{\sigma}
\renewcommand{\th}{\theta}
\newcommand{\vp}{\varphi}

\newcommand{\ze}{\zeta}
\newcommand{\De}{\Delta}
\newcommand{\Ga}{\Gamma}

\newcommand{\tP}{\widetilde{P}}
\newcommand{\tQ}{\widetilde{Q}}
\newcommand{\tR}{\widetilde{R}}

\newcommand{\tc}{\tilde{c}}

\newcommand{\tH}{\widetilde{H}}

\newcommand{\ssC}{\mathsf{C}}

\newcommand{\ssH}{\mathsf{H}}

\newcommand{\hP}{\widehat{P}}

\newcommand{\NN}{{\mathbb N}}
\newcommand{\RR}{{\mathbb R}}
\newcommand{\ZZ}{{\mathbb Z}}
\newcommand{\cL}{{\mathcal L}}
\newcommand{\cM}{{\mathcal M}}
\newcommand{\cN}{{\mathcal N}}
\newcommand{\cP}{{\mathcal P}}

\newcommand{\cS}{{\mathcal S}}

\newcommand{\pd}{\partial}

\newcommand{\ket}[1]{|#1\rangle}
\newcommand{\bra}[1]{\langle#1|}

\newcommand{\mss}{\kern 1pt}

\renewcommand{\le}{\leqslant}
\renewcommand{\ge}{\geqslant}
\newcommand{\tends}[1]{\bbuildrel{\hbox to 2em{\rightarrowfill}}_{#1}^{}}

\newcommand{\operatorname}[1]{\mathop{\rm #1}\nolimits}
\newcommand{\sech}{\operatorname{sech}}
\newcommand{\csch}{\operatorname{csch}}

\newcommand{\arctanh}{\operatorname{arctanh}}
\newcommand{\arcsn}{\operatorname{arcsn}}

\newcommand{\sgn}{\operatorname{sgn}}

\newcommand{\diag}{\operatorname{diag}}

\newcommand{\iu}{\mathrm i}
\newcommand{\diff}{\mathrm{d}}

\newcommand{\sla}{\mathrm{sl}}

\newcommand{\implies}{\Longleftrightarrow}

\newcommand{\sn}{\operatorname{sn}}
\newcommand{\cn}{\operatorname{cn}}

\newcommand{\dn}{\operatorname{dn}}

\newcommand{\en}{\enspace}

\newcommand{\Int}[1]{\,\mathop{\!#1}\limits^{\lower1ex\hbox{$\scriptstyle\circ$}}{}}

\newcommand{\dc}{c^\dagger}

\newtheorem{theorem}{Theorem}

\theoremstyle{remark}
\newtheorem{remark}{Remark}

\let\tfrac\case
\let\eqref\eref
\eqnobysec
\newcommand{\kF}{k_{\mathrm F}}
\newcommand{\psiL}{{\psi\vphantom{\psi^\dagger}}_{\mathrm L}}
\newcommand{\psiR}{{\psi\vphantom{\psi^\dagger}}_{\mathrm R}}
\newcommand{\Lr}{L_{\mathrm{rel}}}
\def\clap#1{\hbox to 0pt{\hss#1\hss}}

\begin{document}

\title{Inhomogeneous XX spin chains and quasi-exactly solvable
  models}

\author{Federico Finkel and Artemio González-López
}

\address{Depto.~de Física Teórica, Facultad de Ciencias Físicas, Plaza de las Ciencias 1,\\
  Universidad Complutense de Madrid, 28040 Madrid, SPAIN}

\eads{\mailto{ffinkel@ucm.es}, \mailto{artemio@ucm.es}}

\vspace{10pt}
\begin{indented}
\item[]July 1, 2020
\end{indented}
\begin{abstract}
  We establish a direct connection between inhomogeneous XX spin chains (or free fermion systems
  with nearest-neighbors hopping) and certain QES models on the line giving rise to a family of
  weakly orthogonal polynomials. We classify all such models and their associated XX chains, which
  include two families related to the Lamé (finite gap) quantum potential on the line. For one of
  these chains, we numerically compute the Rényi bipartite entanglement entropy at half filling
  and derive an asymptotic approximation thereof by studying the model's continuum limit, which
  turns out to describe a massless Dirac fermion on a suitably curved background. We show that the
  leading behavior of the entropy is that of a $c=1$ critical system, although there is a
  subleading $\log(\log N)$ correction (where $N$ is the number of sites) unusual in this type of
  models.
\end{abstract}


\submitto{J.~Stat.~Mech.}

\maketitle


\section{Introduction}
\label{sec.intro}

Inhomogeneous XX spin chains ---or, equivalently, systems of free spinless fermions with
nearest-neighbors hopping--- have recently received considerable attention due to their remarkable
entanglement properties. Indeed, the critical phases of these models are effectively described by
$1+1$ dimensional conformal field theories (CFTs), whose entanglement has been extensively studied
using standard field-theoretic techniques (see, e.g., \cite{CC04JSTAT,CC09}). In particular, the
single-block Rényi entanglement entropy of $1+1$ dimensional CFTs features a characteristic
logarithmic growth with the block length $L$, rather than the usual linear growth of thermodynamic
entropy. Thus it is to be expected that the entanglement entropy of critical inhomogeneous XX
chains also scales proportionally to $\log L$ as $L$ goes to infinity; in particular, the model's
central charge can be inferred from the proportionality constant multiplying $\log L$ in the
asymptotic formula for the entropy. The entanglement entropy of the homogeneous XX chain has in
fact been thoroughly studied in the open, closed and (semi-)infinite cases~\cite{JK04,CE10,FC11},
as well as for subsystems consisting of more than one
block~\cite{CH09,CCT09,ATC10,FC10,AEF14,CFGT17}. In all cases the leading behavior of the Rényi
entropy has been found to be logarithmic, with the expected central charge $c=1$ of a free fermion
CFT.

The asymptotic behavior of the entanglement entropy of truly inhomogeneous XX chains is much
harder to establish, since in this case the correlation matrix is neither Toeplitz nor Toeplitz
plus Hankel, so that the standard techniques based on using proved cases of the Fisher--Hartwig
conjecture~\cite{FH68,Ba79,DIK11} to approximate the characteristic polynomial of the latter
matrix for large $L$ cannot be applied. However, in some cases an approximation of the
entanglement entropy can be found by exploiting the model's connection with a suitable CFT. This
idea has been successfully applied to the so-called rainbow chain~\cite{VRL10}, whose hopping
amplitudes decay exponentially outwards from both sides of the chain's center. More precisely, it
was first shown in Ref.~\cite{RDRCS17} that in the continuum limit the rainbow chain's
Hamiltonian tends to that of a massless Dirac fermion in a suitably curved $1+1$ dimensional
background, whose metric's conformal factor is proportional to the square of the chain's hopping
amplitude. Using the results in Ref.~\cite{DSVC17} for the entanglement entropy of the latter
model, the authors in Ref.~\cite{TRS18} derived an asymptotic formula for the entanglement entropy
of the rainbow chain which is in excellent agreement with the numerical data.

Another key feature of inhomogeneous XX chains is their close connection to the classical theory
of orthogonal polynomials, stemming from the fact that the matrix of the system's single-particle
(fermionic) Hamiltonian in the position basis is real, symmetric and tridiagonal, and thus its
entries can be used to define a three-term recursion relation determining a (finite) orthogonal
polynomial system. In this way a one-to-one correspondence between inhomogeneous XX chains and
orthogonal polynomial families is established. In fact, the zeros of the last polynomial in the
family coincide with the chain's single particle energies, and the chain's complete spectrum is
obtained by exciting an arbitrary number of these single particle states. This close connection
between inhomogeneous XX chains and orthogonal polynomials was used, for instance, to characterize
chains of this type with perfect state transfer~\cite{CJ10,Je11,VZ12} or, more recently, to
construct a tridiagonal matrix commuting with the hopping matrix of the entanglement Hamiltonian
of several inhomogeneous XX chains constructed from well-known families of discrete (finite)
polynomial systems~\cite{CNV19}.

The theory of orthogonal polynomials has also close connections to another class of
one-dimensional (one body) quantum systems, namely quasi-exactly solvable (QES) dynamical models
on the line. In general, these models are characterized by the fact that a subset of the spectrum
can be found through algebraic procedures. The prototypical examples of QES models are those whose
Hamiltonian can be expressed as a quadratic polynomial in the generators of the standard spin
$(N-1)/2$ representation of the $\sla(2)$ algebra in terms of first-order differential
operators~\cite{Tu88,Sh89,ST89,Us94}. After an appropriate (pseudo-)gauge transformation, the
(formal) eigenfunctions of these models can be expressed as power series in a suitable variable
$z$, whose coefficients $P_n(E)$ depend polynomially on the energy $E$~\cite{BD96}. In many cases,
these polynomials satisfy a three-term recursion relation, and are therefore orthogonal with
respect to a suitable measure~\cite{FGR96}. Moreover, for certain values of the parameters in the
Hamiltonian there is a positive integer $N$ such that the polynomials $P_n(E)$ with $n\ge N$ are
divisible by $P_N(E)$. Thus the model becomes QES, as the gauged Hamiltonian obviously admits
polynomial eigenfunctions of degree up to $N-1$ in the variable $z$ with energies equal to the $N$
zeros of the critical polynomial $P_N$. Furthermore, in this case the polynomial family
$\{P_n(E)\}_{n=0}^\infty$ is weakly orthogonal, since it can be shown that the polynomials of
degree greater than or equal to $N$ have zero norm.

The aim of this paper is to classify all the inhomogeneous XX spin chains associated with a QES
model on the line, in the sense that they share the same family of orthogonal polynomials. In
other words, we look for chains whose parameters (hopping amplitudes and on-site energies) are
derived from the coefficients of the three-term recursion relation of the weakly orthogonal
polynomial system associated to a QES model. To this end, we first show that there are exactly six
inequivalent types of QES models on the line giving rise to a weakly orthogonal polynomial system,
up to projective transformations. We then prove that any XX chain with the above property is
isomorphic to one of the six chains constructed from the latter canonical forms. One of the
characteristic properties of these chains is that their hopping amplitudes and on-site energies
are algebraic functions of the site index $n$. Remarkably, this is also the case for the models
recently constructed from classical Krawtchouk and dual Hahn polynomials in
Refs.~\cite{CJ10,Je11,CNV19}, although they all differ from the six new chains introduced in this
work. Among these new models there are, in particular, several chains constructed from different
QES realizations of the celebrated Lamé (finite gap) potential~\cite{Ar64}.

As remarked in Ref.~\cite{CNV19}, from the knowledge of the orthogonal polynomial system defined
by an inhomogeneous XX chain it is straightforward to find an explicit expression for the
corresponding free fermion system's correlation matrix, whose eigenvalues yield the model's
entanglement entropy~\cite{VLRK03,Pe03}. This is in fact a very efficient method for computing the
entanglement entropy (in fact, the whole entanglement spectrum), since it is based on
diagonalizing an $L\times L$ matrix instead of the $2^L\times 2^L$ reduced density matrix. We have
applied this idea to compute the Rényi entanglement entropy of one of the new inhomogeneous XX
spin chains constructed from the Lamé potential, whose on-site energies are all zero. Using the
method developed in Ref.~\cite{TRS18}, we have constructed the continuum limit of the latter
chain, which again describes a massless Dirac fermion in a curved $1+1$ dimensional background. In
this way we have obtained an asymptotic formula for the chain's Rényi entanglement entropy when
the number of sites goes to infinity, which is shown to be in excellent agreement with the
numerical results for up to $600$ sites. In particular, this confirms that the model has a
critical phase with $c=1$, as expected.

This paper is organized as follows. In Section~\ref{sec.prelim} we present the models and discuss
their connection with Jacobi (tridiagonal symmetric) matrices. Section~\ref{sec.OPS} includes a
brief summary of several fundamental results from the classical theory of orthogonal polynomials
of interest in the sequel. In particular, we deduce a closed formula for the weights of a finite
(or weakly orthogonal) polynomial system in terms of the zeros of the critical polynomial. In
Section~\ref{sec.conn} we outline the application of these results to the diagonalization of the
hopping matrix of inhomogeneous XX chains (or free fermion systems), establishing a one-to-one
correspondence between orthogonal polynomial systems and inhomogeneous XX spin chains.
Section~\ref{sec.QES} contains a concise review of QES models on the line constructed from the
$\sla(2)$ algebra, with special emphasis on their connection with weakly orthogonal polynomial
systems. In Section~\ref{sec.class} we present our classification of all inequivalent XX spin
chains constructed from the orthogonal polynomial families determined by QES models on the line.
Section~\ref{sec.ent} is devoted to the study of the Rényi entanglement entropy of one of the new
XX spin chains introduced in the previous section, connected to a QES realization of the Lamé
potential. The paper ends with a technical appendix which provides the complete details of the
classification presented in Section~\ref{sec.class}.

\section{Inhomogeneous XX spin chains}\label{sec.prelim}

Assuming conservation of the total number of fermions, the most general free fermion Hamiltonian
with nearest-neighbors hopping can be written as
\begin{equation}\label{Hgen}
  H=\sum_{n=0}^{N-2}J_n\big(\e^{\iu\al_n}{\hat c}^\dagger_n\hat c_{n+1}+\e^{-\iu\al_n}\hat
  c^\dagger_{n+1}\hat c_n\big) +\sum_{n=0}^{N-1}B_n\hat c^\dagger_n\hat c_n\,,
\end{equation}
where $J_n\ge0$, $\al_n,B_n\in\RR$ and the operators $\{\hat c_n,\hat c^\dagger_n\}_{n=0}^{N-1}$
are a family of fermionic operators satisfying the canonical anticommutation relations (CAR)
\[
  \{\hat c_n,\hat c_m\}=\{\hat c^\dagger_n,\hat c^\dagger_m\}=0\,,\qquad \{\hat c_n,\hat
  c^\dagger_m\}=\de_{nm}\,.
\]
As a matter of fact, the latter Hamiltonian can be brought to a simpler canonical form by
introducing the equivalent family of fermionic operators
\[
  c_n=\e^{\iu\be_n}\hat c_n\,,
\]
with suitably chosen phases $\be_n\in\RR$. Indeed, we clearly have
\[
  H=\sum_{n=0}^{N-2}J_n\big(\e^{\iu\al_n}\e^{\iu(\be_n-\be_{n+1})}\dc_nc_{n+1}
  +\e^{-\iu\al_n}\e^{-\iu(\be_n-\be_{n+1})}\dc_{n+1}c_n\big)
  +\sum_{n=0}^{N-1}B_n\dc_nc_n\,,
\]
so that choosing~$\be_{n+1}-\be_n=\al_n$, i.e.,
\[
  \be_n=\sum_{j=0}^{n-1}\al_j\,,
\]
the original Hamiltonian~\eqref{Hgen} reduces to
\begin{equation}
  \label{Hffs}
  H = \sum_{n=0}^{N-2}J_n(\dc_nc_{n+1}+\dc_{n+1}c_{n})+\sum_{n=0}^{N-1}B_n\dc_nc_n\,.
\end{equation}
This model describes a system of $N$ hopping spinless fermions with real hopping amplitudes $J_n$
and chemical potentials (or on-site energies) $B_n$. The latter Hamiltonian obviously commutes
with the total fermion number operator
\[
  \cN=\sum_{n=0}^{N-1}\dc_nc_n\,,
\]
so that the number of fermions is indeed conserved. In what follows we shall always assume that
the hopping amplitudes do not vanish, so that
\begin{equation}\label{Jnpos}
  J_n>0\,,\qquad 0\le n\le N-2\,.
\end{equation}
Note that, by the previous observation, the latter model is trivially equivalent to the analogous
one with nonvanishing (positive or negative) hopping amplitudes $\vep_nJ_n$ with arbitrary signs
$\vep_n\in\{\pm1\}$.

As is well known, under the Jordan--Wigner transformation
\begin{equation}\label{WJ}
  c_n=\prod_{k=0}^{n-1}\si_k^z\cdot\si_n^{+}\,,\qquad 0\le n\le N-1\,,
\end{equation}
where $\si_n^\al$ (with $\al=x,y,z$) denotes the Pauli matrix $\si^\al$ acting on the $n$-th site
and $\si^{\pm}_\al=(\si^x_\al\pm\iu\si^y_\al)/2$, the Hamiltonian~\eqref{Hffs} is transformed into
the spin $1/2$ open XX chain Hamiltonian
\begin{equation}
  \label{Hchain}
  H = \frac12\sum_{n=0}^{N-2}J_n(\si_n^x\si_{n+1}^x+\si_n^y\si_{n+1}^y)
  +\frac12\sum_{n=0}^{N-1}B_n(1-\si_n^z)\,.
\end{equation}
Thus the models \eqref{Hffs} and~\eqref{Hchain} can be regarded as essentially equivalent. In
particular, the fermionic vacuum~$\ket0$ corresponds to the reference
state~$\ket{\uparrow\cdots\uparrow}$ with all spins up, while the general fermionic state
\[
  \dc_{n_0}\cdots\dc_{n_k}\ket0\,,\qquad 0\le n_0<\cdots<n_k\le N-1\,,
\]
is easily seen to correspond to the state
\[
  \si_{n_0}^-\cdots\si_{n_k}^-\ket{\uparrow\cdots\uparrow}
\]
with flipped spins at positions $n_0<\cdots<n_k$. Note also in this respect that the
transformation
\[
  c_n\mapsto \e^{-\iu\be_n}c_n\,,
\]
with $\be_n$ real, corresponds to
\[
  \si_n^\pm\mapsto\e^{\pm\iu\be_n}\si_n^\pm,\qquad \si_n^z\mapsto\si_n^z\,,
\]
which clearly preserves the commutation relations of the Pauli matrices. In what follows we shall
mostly work with the fermionic model~\eqref{Hffs}, our results being easily translated to the XX
spin chain~\eqref{Hchain} by the previous considerations.

Let~$H_1$ denote the restriction of the fermionic Hamiltonian $H$ in Eq.~\eqref{Hffs} to the
subspace of one-particle states, a basis of which consists of the states
\begin{equation}\label{posbasis}
  \ket n:=\dc_n\ket0\,,\qquad 0\le n\le N-1\,,
\end{equation}
with a single fermion at each site $n$. The matrix~$\ssH=(H_{nm})_{n,m=0}^{N-1}$ of $H_1$ in the
position basis~\eqref{posbasis} has matrix elements
\begin{equation}\label{hdef}
  H_{nm}=\bra nH\ket m=J_{n}\de_{m,n+1}+J_{n-1}\de_{m,n-1}+B_n\de_{nm}\,,
\end{equation}
so that $\ssH$ is the $N\times N$ tridiagonal matrix
\begin{equation}\label{Hmat}
  \ssH=
  \left(
  \begin{array}{cccccccc}
    B_0& J_0& 0& 0& \cdots &0 &0 &0\\
    J_0& B_1& J_1& 0& \cdots &0 &0 &0\\
    0& J_1& B_2& J_2& \cdots &0 &0 &0\\
    \cdot& \cdot& \cdot& \cdot& \cdots &\cdot &\cdot &\cdot\\
    0& 0& 0& 0& \cdots &J_{N-3} & B_{N-2} & J_{N-2}\\
    0& 0& 0& 0& \cdots &0 & J_{N-2} & B_{N-1}
  \end{array}\right)\,.
\end{equation}
Note that in terms of this matrix the full Hamiltonian~$H$ can be expressed in matrix notation as
\begin{equation}\label{HC}
  H=\ssC^\dagger \ssH \ssC\,,
\end{equation}
where $\ssC=(c_0\dots c_{N-1})^T$ and $\ssC^\dagger=(\dc_0\dots\dc_{N-1})$. Since $\ssH$ is real
symmetric, it can be diagonalized by means of a real orthogonal transformation~$\Phi$,
i.e.,
\begin{equation}\label{HPhi}
  \Phi^T\ssH\Phi=\diag(E_0,\dots,E_{N-1})\,,
\end{equation}
where $E_0\le\cdots\le E_{N-1}$ are the (real) eigenvalues of $\ssH$.
Let
\[
  \Phi_{nk}=:\phi_n(E_k)\,,
\]
an define a new set of fermionic operators~$\tc_n$ by
\begin{equation}\label{mommodes}
  \tc_k:=\sum_{n=0}^{N-1}\phi_n(E_k)c_n\,,\qquad 0\le k\le N-1,
\end{equation}
which satisfy the CAR on account of the orthogonal character of $\Phi$. Since
$\widetilde\ssC=\Phi^T\ssC$, from Eqs.~\eqref{HC}-\eqref{HPhi} it follows that
\begin{equation}\label{Hdiag}
  H=\widetilde\ssC^\dagger (\Phi^T\ssH\Phi) \widetilde\ssC=\sum_{n=0}^{N-1}E_k\tc^\dagger_n\tc_n\,.
\end{equation}
Thus the \emph{full} Hamiltonian $H$ is diagonal in the basis consisting of the
states
\begin{equation}\label{mombasis}
  \tc^\dagger_{n_0}\cdots\tc_{n_k}^\dagger\ket0\,,\qquad 0\le n_0<\cdots<n_k\le N-1\,,
\end{equation}
whose corresponding energy is given by
\begin{equation}\label{spectrum}
  E(n_0,\dots,n_k)=\sum_{j=0}^{k}E_{n_j}\,.
\end{equation}
In particular, the states $\tc^\dagger_k\ket0$ (with $0\le k\le N-1$) are single-fermion
excitation modes with energy $E_k$. Note, finally, that Eq.~\eqref{HPhi} is equivalent to the
system of $N$ equations
\[
  \sum_{m=0}^{N-1}H_{nm}\phi_m(E_k)=E_k\phi_n(E_k)\,,
\]
or, taking~\eqref{hdef} into account,
\begin{equation}
  \label{phieqs}
  \fl
  E_k\phi_n(E_k)=J_n\phi_{n+1}(E_k)+B_n\phi_n(E_k)+J_{n-1}\phi_{n-1}(E_k)\,,\qquad
  0\le n\le N-1\,,
\end{equation}
with $J_{-1}=J_{N-1}=0$\,. More precisely, the first $N-1$ equations~\eqref{phieqs} determine
$\phi_n(E_k)$ with $n=1,\dots,N-2$ up to the proportionality factor $\phi_0(E_k)\ne0$. The last
equation, which on account of the condition $J_{N-1}=0$ reads
\begin{equation}\label{lasteq}
  (E_k-B_{N-1})\phi_{N-1}(E_k)-J_{N-2}\phi_{N-2}(E_k)=0\,,
\end{equation}
then yields a polynomial equation of degree $N$ in $E_k$ which determines the $N$ single-fermion
excitation energies $E_k$. Finally, the factor $\phi_0(E_k)$ is determined (up to a sign) imposing
the orthonormality condition
\[
  \sum_{n=0}^{N-1}\phi_n^2(E_k)=1\,.
\]
Note that the full orthogonality conditions
\[
  \sum_{n=0}^{N-1}\phi_n(E_j)\phi_n(E_k)=\de_{jk}\,,\qquad 0\le j,k,\le N-1,
\]
or equivalently
\begin{equation}\label{ortphi}
  \sum_{k=0}^{N-1}\phi_n(E_k)\phi_m(E_k)=\de_{nm}\,.\qquad 0\le n,m\le N-1,
\end{equation}
are then automatically satisfied if the eigenvalues $E_k$ of $\ssH$ are simple. In fact, we shall
show below that this is guaranteed by conditions~\eqref{Jnpos}.

\section{Orthogonal polynomials}\label{sec.OPS}

Equations~\eqref{phieqs} determining the matrix elements $\phi_n(E_k)$ (up to normalization) are
reminiscent of the three-term recurrence relation satisfied by a \emph{finite} orthogonal
polynomial system (OPS) $\{P_n(E):n=0,\dots,N\}$. More precisely, taking $P_n$ to be monic for all
$n$ the recurrence relation satisfied by such a system can be written as
\begin{equation}\label{rr}
  P_{n+1}(E)=(E-b_n)P_n-a_nP_{n-1}\,,\qquad 0\le n\le N-1\,,
\end{equation}
where $a_0:=0$ and $P_0(E):=1$. We shall only assume in what follows that
\begin{equation}
  \label{condsab}
  a_n>0 \en\text{for}\en n=1,\dots,N-1\,,\qquad b_n\in\RR\en\text{for}\en n=0,\dots N-1\,.
\end{equation}
In many cases such a finite system is obtained by truncating an infinite orthogonal (or weakly
orthogonal) polynomial family $\{P_n(E):n=0,1,\dots\}$, but this need not be the case.
\begin{theorem}
  The zeros of each polynomial $P_n$ with $1\le n\le N$ are real and simple.
\end{theorem}
\begin{proof}
  Extend the finite OPS to an infinite one by arbitrarily defining $a_n>0$ and $b_n\in\RR$ for all
  $n\ge N$. By Favard's theorem, there is a (unique) positive definite moment functional with
  support on the whole real line with respect to which all the polynomials of the extended OPS are
  mutually orthogonal. By Theorem 5.2 of Ref.~\cite{Ch78}, the zeros of all polynomials in the
  family ---and, in particular, of $P_1,\dots,P_N$--- are real and simple.
\end{proof}

In view of the latter theorem, let us denote by $E_0<\cdots<E_{N-1}$ the $N$ (real) roots of the
last polynomial $P_N$ in the family. We shall next construct a positive definite discrete moment
functional $\cL={}$\hbox{$\sum_{k=0}^{N-1}w_k\de(E-E_k)$} supported on these zeros, with respect
to which the polynomials $P_n$ with $0\le n\le N-1$ are mutually orthogonal, i.e,
\begin{equation}\label{Pnmorth}
  \fl
  \bigl\langle P_n,P_m\bigr\rangle:=\cL(P_nP_m)=\sum_{k=0}^{N-1}w_kP_n(E_k)P_m(E_k)=0\,,
  \qquad 0\le n\ne m\le N-1\,.
\end{equation}
In fact, we shall show that this functional is unique if we set (as is customary)
$\langle P_0,P_0\rangle=1$. Indeed, extend again the given finite OPS to an infinite one with
$a_n>0$ and $b_n$ real for all $n$. As mentioned above, there is a positive definite linear
functional $L$ with respect to which the polynomials of the infinite family are mutually
orthogonal, which is unique if we impose the normalization condition $\mu_0:=L(1)=1$. By
Theorem~6.1 of Ref.~\cite{Ch78}, there are positive weights $w_k$ (with $0\le k\le N-1$) such that
the restriction of $L$ to the space of polynomials $p(x)$ of degree not greater than $2N-1$ is
given by
\[
  \cL(p)=\sum_{k=0}^{N-1}w_kp(E_k)\,.
\]
In particular, Eq.~\eqref{Pnmorth} holds for this moment functional.

Once this result is established, we can easily find the square
norm~$\ga_n:=\langle P_n,P_n\rangle$ of each polynomial $P_n$ with $n=0,\dots,N-1$ and the weights
$w_k$. Indeed, taking the scalar product of the recurrence relation~\eqref{rr} with the polynomial
$P_{n-1}$ (with $1\le n\le N-1$) we obtain
\[
  0=\langle EP_{n-1},P_n\rangle-a_n\ga_{n-1}=\ga_n-a_n\ga_{n-1}\,.
\]
Assuming (as shall be done in the sequel) that $\ga_0=\cL(1)=1$ we obtain the formula
\begin{equation}\label{gan}
  \ga_n=\prod_{k=1}^na_n\,,\qquad 0\le n\le N-1\,. 
\end{equation}
We can thus write
\begin{equation}
  \label{orthPs}
  \sum_{k=0}^{N-1}w_kP_n(E_k)P_m(E_k)=\ga_n\de_{nm}\,,
  \qquad 0\le n,m\le N-1\,,
\end{equation}
with $\ga_n$ given by Eq.~\eqref{gan}. Secondly, the simple character of the roots of the
polynomial~$P_N$ entails the following explicit formula for the weights $w_k$:
\begin{equation}
  \label{wkexp}
  w_k=\frac{\prod_{n=1}^{N-1}a_n}{P_{N-1}(E_k)P_N'(E_k)}\,,\qquad 0\le k\le N-1\,.
\end{equation}
Indeed, since
\[
  \pi_k(E):=\frac{P_N(E)}{E-E_k}=\prod_{n=0\atop n\ne k}^{N-1}(E-E_n)
\]
is, like $P_{N-1}$, a monic polynomial of degree $N-1$, we have
\[
  \langle \pi_k,P_{N-1}\rangle=\ga_{N-1}=
  w_kP_{N-1}(E_k)\prod_{n=0\atop n\ne k}^{N-1}(E_k-E_n)=w_kP_{N-1}(E_k)P_N'(E_k).
\]
In particular, Eq.~\eqref{wkexp} shows that the weights $w_k$ are uniquely determined. Note also
that $\sgn P_N'(E_k)=\sgn P_{N-1}(E_k)=(-1)^{N-k-1}$, since $E_k$ lies between the $k$-th and the
$(k+1)$-th zero of $P_{N-1}$ on account of the interlacing theorem~\cite{Ch78}. Thus $w_k>0$ for
all $k=0,\dots,N-1$, as it should.

The previous considerations can be summarized in the following theorem:
\begin{theorem}\label{thm.OPS}
  Let $\{P_n:n=0,\dots,N\}$ be a finite OPS defined by the recursion relation~\eqref{rr}, with
  coefficients $a_n,b_n$ satisfying conditions~\eqref{condsab}. Then the orthogonality
  conditions~\eqref{orthPs} hold, where the positive weights $w_k$ (with $k=0,\dots,N-1$) are
  defined by Eq.~\eqref{wkexp} and $\ga_n>0$ is given by Eq.~\eqref{gan}.
\end{theorem}

\section{Connection between finite OPSs and inhomogeneous XX chains}\label{sec.conn}

Comparing the orthogonality relations~\eqref{ortphi} and~\eqref{orthPs} immediately suggests a
connection between a finite OPS satisfying conditions~\eqref{condsab} and an inhomogeneous XX
chain~\eqref{Hchain} or free fermion system~\eqref{Hffs}. Indeed, the orthogonality
conditions~\eqref{ortphi} for the matrix elements $\phi_n(E_k)$ will automatically hold provided
that
\begin{equation}
  \label{phiP}
  \phi_n(E_k)=\sqrt{\frac{w_k}{\ga_n}}\,P_n(E_k)\,,
  \qquad 0\le k,n\le N-1\,.
\end{equation}
Note that the right-hand side of the latter equation is real and well defined, since by
Theorem~\ref{thm.OPS} conditions~\eqref{condsab} guarantee that $w_k,\ga_n>0$ for all
$k,n=0,\dots, N-1$. To find the couplings $J_n$ and the magnetic field strengths $B_n$, we combine
the recursion relation~\eqref{rr} with the Eq.~\eqref{phiP}, thus obtaining
\[
  \fl
  \sqrt{\ga_{n+1}}\,\phi_{n+1}=(E_k-b_n)\sqrt{\ga_n}\,\phi_n(E_k)
  -\sqrt{\ga_{n-1}}\,a_n\phi_{n-1}(E_k)\,,
  \qquad 0\le n\le N-2\,,
\]
or, taking into account Eq.~\eqref{gan} for $\ga_n$,
\[
  \sqrt{a_{n+1}}\,\phi_{n+1}(E_k)=(E_k-b_n)\phi_n(E_k)
  -\sqrt{a_{n}}\,\phi_{n-1}(E_k)\,,\qquad 0\le n\le N-2\,.
\]
Comparing with Eqs.~\eqref{phieqs} we immediately arrive at the relations
\[
  J_n=\sqrt{a_{n+1}}\,,\quad B_n=b_n\,,\qquad  0\le n\le N-2\,.
\]
We still have to enforce the recursion relation~\eqref{rr} for $n=N-1$, which taking into account
the identity $P_N(E_k)=0$ and the previous relations yields
\[
  \fl
  0=(E_k-b_{N-1})\phi_{N-1}(E_k)-\sqrt{a_{N-1}}\phi_{N-2}(E_k)\\
  =(E_k-b_{N-1})\phi_{N-1}(E_k)-J_{N-2}\phi_{N-2}(E_k).
\]
Comparing with Eq.~\eqref{lasteq} we thus conclude that $B_{N-1}=b_{N-1}$. In summary, we have
established the following result:

\begin{theorem}\label{thm.main}
  Let $\{P_n:n=0,\dots,N\}$ be a finite OPS defined by the recursion relation~\eqref{rr}, with
  coefficients $a_n>0$ and $b_n\in\RR$. Then the inhomogeneous open XX chain with
  Hamiltonian~\eqref{Hchain} ---or, equivalently, the free fermion system with
  Hamiltonian~\eqref{Hffs}--- and coefficients
  \begin{equation}\label{JBn}
    J_n=\sqrt{a_{n+1}}\,,\qquad B_n=b_n
  \end{equation}
  is diagonal in the basis~\eqref{mommodes}-\eqref{mombasis}, where the single-fermion excitation
  energies $E_0<\cdots<E_{N-1}$ are the zeros of the polynomial $P_N$ and the coefficients
  $\phi_n(E_k)$ in Eq.~\eqref{mommodes} are given by Eq.~\eqref{phiP}.
\end{theorem}
\begin{remark}
  Of course, Theorem~\ref{thm.main} can be easily extended to the (apparently more general)
  Hamiltonian~\eqref{Hgen} and, in particular, to the model~\eqref{Hffs} with nonzero hopping
  amplitudes $\vep_nJ_n$ with arbitrary signs $\vep_n$. More precisely, the single-fermion
  excitation energies of the latter Hamiltonian are still the roots of the last polynomial $P_N$
  from the finite OPS satisfying \eqref{rr}-\eqref{condsab}, and~\eqref{Hgen} can be brought into
  the diagonal form~\eqref{Hdiag} introducing the operators~\eqref{mommodes} with
  \[
    \phi_n(E_k)=\e^{-\iu\sum_{l=0}^{n-1}\al_l}\sqrt{\frac{w_k}{\ga_n}}\,P_n(E_k)\,, \qquad 0\le
    k,n\le N-1\,.
  \]
  Note that the matrix $\Phi$ with matrix elements $\phi_{n}(E_k)$ (with $k,n=0,\dots,N-1$) is
  still unitary, as a result of the orthogonality condition~\eqref{orthPs} satisfied by the
  polynomials $P_n$. Thus the operators $\{\tc_n,\tc_n^\dagger:n=0,\dots,N-1\}$ satisfy the CAR,
  with $\tc_n^\dagger$ creating the $n$-th single-fermion energy eigenstate. This implies that the
  spectrum of the model~\eqref{Hgen} depends only on $J_n$ and $B_n$, a fact which obviously also
  follows from the observation at the beginning of Section~\ref{sec.prelim}. In particular, the
  $2^{N-1}$ XX chains with $J_n>0$ and hopping amplitudes $\vep_nJ_n$ with arbitrary signs
  $\vep_n$ are \emph{isospectral}. In other words, the spectrum of the chain~\eqref{Hchain} with
  $J_n$ real and nonvanishing depends only on $|J_n|$ and $B_n$ only. It immediately follows from
  this observation that the spectrum of the XX chain~\eqref{Hchain} with $B_n=0$ for
  $n=0,\dots,N-1$ is symmetric about $0$.
\end{remark}
\begin{remark}
  The fact that the spectrum of the Hermitian tridiagonal matrix
  \[
    \left(
    \begin{array}{cccccccc}
      B_0& J_0 & 0& \cdots &0 &0 &0& 0\\
      J_0^*& B_1& J_1& 0& \cdots &0 &0 &0\\
      0& J_1^*& B_2& J_2& \cdots &0 &0 &0\\
      \cdot& \cdot& \cdot& \cdot& \cdots &\cdot &\cdot &\cdot\\
      0& 0& 0& 0& \cdots &J_{N-3}^* & B_{N-2} & J_{N-2}\\
      0& 0& 0& 0& \cdots &0 & J_{N-2}^* & B_{N-1}
    \end{array}
  \right)
\]
depends only on $|J_n|$ and $B_n$, and that its eigenvalues are the roots of the polynomial $P_N$
(provided that $|J_n|=\sqrt{a_{n+1}}$ and $B_n=b_n$), is well known in the classical theory of
orthogonal polynomials (cf.~Exercise 5.7 in Ref.~\cite{Ch78}).
\end{remark}

In practice, Theorem~\ref{thm.main} is usually applied in two somewhat different situations:

\begin{enumerate}[I.]
\item $\{P_n:n=0,1,\dots,\}$ is an infinite polynomial family determined by the recursion
  relation~\eqref{rr} with coefficients $a_n>0$, $b_n\in\RR$ \emph{independent of} $N$.
\item $\{P_n:n=0,1,\dots,N\}$ is a finite OPS defined by the recursion relation~\eqref{rr}, with
  coefficients $a_n(N)>0$, $b_n(N)\in\RR$ \emph{depending on a positive integer parameter} $N$.
\end{enumerate}
In the first scenario, for each $N\in\NN$ we simply truncate the infinite family to obtain a
finite OPS $\{P_n:n=0,\dots,N\}$ yielding a class of inhomogeneous chains with $N$ sites and
Hamiltonian~\eqref{Hchain} ---equivalently, a class of inhomogeneous free $N$-fermion systems with
Hamiltonian~\eqref{Hffs}--- with coefficients $J_n$, $B_n$ given by~\eqref{JBn} and thus
\emph{independent of $N$}. In particular, the first $N-1$ couplings $J_n$ and $N$ on-site energies
$B_n$ of the chain with $N+1$ sites coincide with those of the corresponding chain with $N$ sites.
Typical instances of this situation are the families of classical orthogonal polynomials. Note
that, by Favard's theorem, in this case there is a unique positive definite moment functional $L$
for the whole infinite OPS. This moment functional is usually (but not always) defined by a
continuous Stieltjes measure $\mu(x)\diff x$, i.e.,
\[
  L(f)=\int_If(x)\mu(x)\diff x\,,
\]
with $I\subset\RR$ a finite or infinite interval (this is the case, for instance, with the
families of classical polynomials). By Favard's theorem and Theorem 6.1 in Ref.~\cite{Ch78}, the
($N$-dependent) moment functional $\cL$ in Eq.~\eqref{Pnmorth} is the restriction of $\mu_0^{-1}L$
to polynomials of degree up to $2N-1$. For other infinite polynomial families (for instance,
Charlier and Meixner polynomials), the moment functional is discrete but infinite, i.e., of the
form
\[
  L(f)=\sum_{k=0}^\infty \om_kf(x_k)\,.
\]
Again, for each $N\in\NN$ the restriction of $\mu_0^{-1}L$ to polynomials of degree up to $2N-1$
coincides with the functional $\cL$ in Theorem~\ref{thm.main}. Moreover, in this case a
straightforward generalization of the method summarized in Theorem~\ref{thm.main} can in principle
be applied to the semiinfinite chain
\[
  H=\sum_{n=0}^\infty J_n(\si_n^x\si_{n+1}^x+\si_n^y\si_{n+1}^y)
  +\frac12\sum_{n=0}^{\infty}B_n(1-\si_n^z)\,,
\]
or equivalently to the semiinfinite free fermion system
\[
  H = \sum_{n=0}^{\infty}J_n(\dc_nc_{n+1}+\dc_{n+1}c_{n})+\sum_{n=0}^{\infty}B_n\dc_nc_n\,,
\]
with coefficients satisfying~\eqref{JBn}. In particular, the relevant moment functional for
this infinite chain or free fermion system is $L$.

Similarly, a typical example of the second scenario described above is the case of an infinite
family of \emph{weakly orthogonal} polynomials~$\{P_n:n=0,1\dots,\}$, determined by a recursion
relation of the form~\eqref{rr}, with coefficients $a_n(N)$, $b_n(N)\in\RR$ depending on a
positive integer parameter $N$ satisfying
\[
  a_n(N)>0,\quad 0\le n\le N-1\,,\qquad a_N(N)=0\,.
\]
For each $N\in\NN$, we can apply Theorem~\ref{thm.main} to the truncated family
$\{P_n:n=0,\dots,N\}$, with weights $w_k$ in Eq.~\eqref{wkexp}, single-fermion excitation energies
$E_k$ and coefficients $J_n$, $B_n$ in the corresponding Hamiltonians~\eqref{Hffs}-\eqref{Hchain}
depending on $N$. An example of this are the well-known families of Krawtchouk and Hahn
polynomials, as well as the families of orthogonal polynomials associated to certain quasi-exactly
solvable one-dimensional quantum models studied in the following sections. Note that, again by
Favard's theorem, in the case of an infinite family of weakly orthogonal polynomials the
functional~$\cL$ in Eq.~\eqref{Pnmorth} is actually the moment functional for the whole infinite
family with parameter $N$, although this result is of little practical use since the polynomials
with index greater than $N$ are multiples of $P_N$.

\section{Quasi-exactly solvable models on the line}\label{sec.QES}

In this section we shall summarize the main facts about quasi-exactly solvable models needed in
the sequel. Roughly speaking, a quasi-exactly solvable model is a quantum system for which part of
the spectrum (although in general not all of it) can be computed algebraically, i.e., by solving a
polynomial equation (or equations) of finite degree. Although this can be achieved in many ways,
the class of QES models of interest in this paper is the family of one-dimensional single-particle
quantum Hamiltonians classified in Refs.~\cite{Tu88,GKO93}. The partial solvability of these
models hinges on their close connection with the $\mathrm{sl}(2)$ algebra spanned by the
first-order differential operators
\begin{equation}\label{sl2}
  L_-=\pd_z,\qquad L_0=z\pd_z-\frac{N-1}2\,,\qquad L_+=z^2\pd_z-(N-1)z\,,
\end{equation}
where $N$ is a positive integer equal to the number of algebraic levels (see, e.g.,
Refs.~\cite{Sh89,GKO94} for reviews on the subject). More precisely, the key assumption is that
the one-dimensional Hamiltonian (in appropriate units)
\begin{equation}\label{QESmodel}
  H=-\pd_x^2+V(x)
\end{equation}
can be mapped by a change of variables $z=\ze(x)$ and a pseudo-gauge
transformation~$H\mapsto H_g=\mu(z)^{-1}H\mu(z)$ to a second-order differential operator $H_g$ in
the variable $z$ ---the so called \emph{gauge Hamiltonian}--- of the form
\begin{equation}\label{Hg}
  H_g=-\sum_{a,b}h_{ab}L_aL_b-\sum_ah_aL_a-h_*\,,
\end{equation}
with $h_{ab}=h_{ba}$, $h_a$ ($a,b=-,0,+$) and $h_*$ real parameters. The gauge factor~$\mu(z)$,
the change of variables~$z=\ze(x)$ and the potential~$V(x)$ can be easily expressed in terms of
these parameters, namely~\cite{FGR96}
\begin{equation}\label{xmu}
  x=\pm\int\frac{\diff z}{\sqrt{P(z)}}\,,\qquad
  \mu(z)=P(z)^{-\frac{N-1}4}\exp\bigg(\int\frac{Q(z)}{2P(z)}\,\diff z\bigg)\,
\end{equation}
and
\begin{equation}\label{V}
  V=(N^2-1)\bigg(\frac{P'^2}{16P}-\frac{P''}{12}\bigg)
  +N\bigg(\frac{Q'}2-\frac{QP'}{4P}\bigg)+\frac{Q^2}{4P}-R\,,
\end{equation}
where $P(Z)$, $Q(z)$, $R$ are the polynomials defined by\footnote{We have set without loss of
  generality $c_{-+}=c_{+-}=0$, on account of the Casimir identity
  $L_0^2-\frac12(L_+L_-+L_-L_+)=\frac14(N^2-1)$.}
\begin{eqnarray}
  \fl
  P(z)=h_{++}z^4+2h_{0+}z^3+h_{00}z^2+2h_{0-}z+h_{--}\,,\qquad
  Q(z)=h_+z^2+h_0z+h_-\,,\nonumber\\
  \fl
  R=h_*+\frac{h_{00}}{12}\,(N^2-1)\,,
    \label{PQR}
\end{eqnarray}
the prime denoting derivative with respect to $z$. In terms of the polynomials $(P,Q,R)$ defined
in the latter equation, the gauge operator~$H_g$ is given by
\begin{eqnarray}
  \fl
  H_g=-P(z)\pd_z^2-\bigg[Q(z)-\frac{N-2}2\,&P'(z)\bigg]\pd_z\nonumber\\
  &-\bigg[R-\frac{N-1}2\,Q'(z)
    +\frac{(N-1)(N-2)}{12}\,P''(z)\bigg].
    \label{HgPQR}
\end{eqnarray}

Since the generators $L_a$ obviously preserve the space $\cP_{N-1}$ of polynomials in~$z$ of
degree less than or equal to $N-1$, the same will be true for the gauge Hamiltonian $H_g$. It
follows from this simple observation that $N$ eigenvalues~$E_k$ (with $0\le k\le N-1$) of $H_g$
and their corresponding eigenfunctions $\vp_k(z)$ can be algebraically computed by diagonalizing
the restriction of~$H_g$ to the finite-dimensional space $\cP_{N-1}$. In turn, this means that the
physical Hamiltonian $H=\mu(z)H_g\mu(z)^{-1}$ possesses $N$ eigenfunctions $\psi_k=\mu\vp_k$ whose
energies $E_k$ can be algebraically determined\footnote{We are, of course, sidestepping several
  technical issues duly addressed in the previously cited references, like for instance the fact
  that $H_g$ is diagonalizable in $\cP_{N-1}$ and that its eigenvalues are real.}. To see this in
more detail, let us look for analytic solutions
\[
  \vp(z)=\sum_{n\ge0}\hP_n(E)\,\frac{z^n}{n!}\,,
\]
of the eigenvalue equation $H_g\vp=E\vp$. It is easy to check from the explicit expression of the
operators $L_a$ that if we set $\hP_0(E)=1$ then $\hP_n(E)$ is a polynomial in $E$. Moreover, if
$E=E_k$ is one of the algebraic eigenvalues of $H_g$ the latter expression for $\vp$ must reduce
to a polynomial of degree less than or equal to $N-1$ in $z$. We must thus have $\hP_n(E_k)=0$ for
$n\ge N$; in other words, the algebraic eigenvalues are the roots of the critical polynomial
$\hP_N$, and $\hP_n=Q_n\hP_N$ (with $Q_n(E)$ a polynomial) for $n\ge N$. The polynomials $\hP_n$
satisfy in general a five-term recursion relation, which reduces to a three-term one provided that
$h_{++}=h_{--}=0$~\cite{FGR96}. More precisely, if the latter conditions hold we have
\begin{eqnarray}\label{rrhP}
  \fl
  \big[(N-2n-2)h_{0-}-h_-\big]\hP_{n+1}=\bigg[E+h_*+&h_0\Big(n-\tfrac{N-1}2\Big)\nonumber
                                                 +h_{00}\Big(n-\tfrac{N-1}2\Big)^2\bigg]\hP_n\\
  \fl
  &-n(N-n)\big[(2n-N)h_{0+}+h_+\big]\hP_{n-1},
\end{eqnarray}
with $\hP_{-1}=0$ and $\hP_0=1$. Thus a QES model with $h_{++}=h_{--}=0$ defines a polynomial
family $\{\hP_n(E):n=0,1,\dots\}$, with $\deg\hP_n=n$, provided that
\begin{equation}
  \label{hPncond}
  A_{n+1}:=(N-2n-2)h_{0-}-h_-\ne0\,,\qquad n=0,1,\dots\,.
\end{equation}
If this is the case the polynomial system defined by~\eqref{rrhP} is always \emph{weakly
  orthogonal}, since the coefficient of $\hP_{n-1}$ in Eq.~\eqref{rrhP} vanishes for $n=N$. This
is consistent with the previous discussion, since $\hP_N(E)=0$ must imply that $\hP_n(E)=0$ for
$n>N$.

In order to relate the results on QES models we have just outlined to those in the previous
sections, assuming that condition~\eqref{hPncond} holds we introduce the monic polynomials
\begin{equation}
  \label{PnQES}
  P_n(E)=\prod_{j=1}^nA_j \cdot\hP_n(E)=\prod_{j=1}^n[(N-2j)h_{0-}-h_-\big]\cdot \hP_n(E)\,,
\end{equation}
which on account of~\eqref{rrhP} satisfy the recursion relation
\begin{eqnarray}
  \label{rrQES}
  \fl
  P_{n+1}=\bigg[E+h_*+h_0&\Big(n-\tfrac{N-1}2\Big)
                  +h_{00}\Big(n-\tfrac{N-1}2\Big)^2\bigg]P_n\nonumber\\
  \fl
  &-n(N-n)\big[(2n-N)h_{0+}+h_+\big]\big[(N-2n)h_{0-}-h_-\big]P_{n-1}
\end{eqnarray}
with $P_{-1}=0$, $P_0=1$. We thus see that a QES model with
\begin{equation}
  \label{condP}
  h_{++}=h_{--}=0
\end{equation}
for which condition~\eqref{hPncond} holds is related in the manner specified by
Theorem~\ref{thm.main} to an inhomogeneous XX chain~\eqref{Hchain} ---or a free fermion
system~\eqref{Hffs}--- with parameters
\begin{equation}
  \label{QESchain}
  \fl
  \eqalign{
  J_n&=\sqrt{(n+1)(N-n-1)\big[(2n+2-N)h_{0+}+h_+\big]\big[(N-2n-2)h_{0-}-h_-\big]}\,,\cr
  B_n&=-h_*-h_0\Big(n-\tfrac{N-1}2\Big)
  -h_{00}\Big(n-\tfrac{N-1}2\Big)^2\,,
  }
\end{equation}
provided that
\begin{equation}
  \label{condan}
  \fl
  \big[(2n+2-N)h_{0+}+h_+\big]\big[(2n+2-N)h_{0-}+h_-\big]<0\,,\qquad 0\le n\le N-2\,.
\end{equation}
Note, in particular, that this condition implies the validity of Eq.~\eqref{hPncond} for
$0\le n\le N-2$. Thus, if Eqs.~\eqref{condan} and~\eqref{hPncond} hold the zeros $E_k$ (with
$0\le k\le N-1$) of the critical polynomial $P_N$ are both the algebraic eigenvalues of the
one-dimensional QES model with potential given by Eq.~\eqref{V} and the single-fermion excitation
energies of the free fermion system with parameters~\eqref{QESchain} ---or, equivalently, the
one-magnon energies of the corresponding chain~\eqref{Hchain}. Moreover, the polynomials $P_n(E)$
evaluated at the energies $E_k$ determine both the single-fermion excitation
operators~$\tc^\dagger_n$ through Eqs.~\eqref{mommodes}-\eqref{phiP} and the algebraic
eigenfunctions~$\psi_k(x)$ of the QES model~\eqref{QESmodel}-\eqref{V} through the formula
\begin{equation}\label{algeig}
  \psi_k(x)=\mu(z)\sum_{n=0}^{N-1}\prod_{j=1}^n[(N-2j)h_{0-}-h_-\big]^{-1}\cdot P_n(E)\,\frac{z^n}{n!}\,,
\end{equation}
with $z$ and $\mu$ given by Eq.~\eqref{xmu}. Note that, as remarked above for the family
$\{\hP_n(E):n=0,1,\dots\}$, the polynomials~$P_n(E)$ with $n=0,1,\dots$ determined by the
recursion relation~\eqref{rrQES} are always weakly orthogonal, since the coefficient $a_n$ in the
latter relation necessarily vanishes for $n=N$.

\section{Classification}\label{sec.class}

\subsection{Preliminaries}

In this section we shall perform an exhaustive classification of all inequivalent XX inhomogeneous
chains that can be constructed from a QES model in the manner explained above. The key idea in
this respect is to take advantage of the action on the gauge Hamiltonian $H_g$ of the group of
projective (Möbius) transformations, mapping a polynomial $p(z)\in\cP_{N-1}$ to the polynomial
\[
  \tilde p(w):=(\ga w+\de)^{N-1}p\biggl(\frac{\al w+\be}{\ga w+\de}\biggr)\in\cP_{N-1}\,,
  \qquad \De:=\al\de-\be\ga\ne0\,.
\]
Indeed~\cite{GKO93}, under the latter transformation $H_g$ is mapped to the operator
\begin{equation}\label{tHg}
  \widetilde H_g=(\ga w+\de)^{N-1}H_g\Big|_{z=\frac{\al w+\be}{\ga w+\de}}(\ga w+\de)^{-(N-1)}\,,
\end{equation}
which is still of the form~\eqref{HgPQR} with $z$ replaced by $w$ and $(P,Q,R)$ replaced by
\begin{equation}\label{tPQR}
  \fl
  \widetilde P(w)=\frac{(\ga w+\de)^{4}}{\De^2}\,P\bigl(\tfrac{\al w+\be}{\ga
    w+\de}\bigr)\,,\qquad \widetilde Q(w)=\frac{(\ga w+\de)^{2}}\De\,Q\bigl(\tfrac{\al w+\be}{\ga
    w+\de}\bigr) \,,\qquad \tR=R\,.
\end{equation}
Moreover, it can be readily verified that the quantities
\[
  \frac{3P'^2}{P}-4P''\,,\quad 2Q'-\frac{QP'}P\,,\quad \frac{Q^2}P
\]
are invariant under the mapping~\eqref{tPQR}, so that the QES potential $V(x)$ given by
Eq.~\eqref{V} is also invariant, i.e., it can be computed either from the triple $(P,Q,R)$ or from
$(\tP,\tQ,\tR)$ obtaining exactly the same result. Since
$\vp_k(z):=\sum_{n=0}^{N-1}\hP_n(E_k)z^n/n!$ is a polynomial eigenfunction of $H_g$ with
eigenvalue $E_k$, Eq.~\eqref{tHg} implies that
\[
  \widetilde\vp_k(w):=(\ga w+\de)^{N-1}\vp_k(z)=\sum_{n=0}^{N-1}\frac{\hP_n(E_k)}{n!}\, (\al
  w+\be)^n(\ga w+\de)^{N-1-n}
\]
is a polynomial eigenfunction of $\tH_g$ with the same eigenvalue. Hence
$\widetilde\psi_k(x)=\widetilde\mu\widetilde\vp_k(w)$ is the corresponding algebraic eigenfunction
of $H$, also with energy $E_k$. In fact, from Eq.~\eqref{xmu} and the invariance of
$(Q/P\,)\diff z$ it follows that (up to a trivial constant factor)
\[
  \widetilde\mu=(\ga w+\de)^{-(N-1)}\mu\,,
\]
and thus
\[
  \widetilde\psi_k=\widetilde\mu\widetilde\vp_k=(\ga w+\de)^{N-1}\widetilde\mu\vp_k=
  \mu\vp_k=\psi_k\,.
\]
Hence the algebraic eigenfunctions of $H$ computed from $\tH_g$ coincide with those obtained from
$H_g$. If we now assume that $\tP(z)$ satisfies the analogue of conditions~\eqref{condP} then we
can write
\[
  \widetilde\vp_k(w)=\sum_{n=0}^{N-1}\widehat{\widetilde P_n}(E_k)\,\frac{w^n}{n!}\,,
\]
where $\widehat{\widetilde P}_n(E)$ is a polynomial in $E$ of degree $n$ satisfying a three-term
recursion relation akin to~\eqref{rrhP}. The corresponding monic polynomial
family~$\{\tP_n(E):0\le n\le N\}$ in general differs from~$\{P_n(E):0\le n\le N\}$ and, as a
consequence, the XX spin chains determined by these families will generally have different
coefficients. Crucially, though, since the algebraic eigenfunctions constructed from both families
have the same energies $E_k$, the polynomials $P_N$ and $\tP_N$ must coincide. Moreover, since
these energies determine the full spectrum of the associated XX chain through
Eq.~\eqref{spectrum}, the chains defined by the polynomial families $\{P_n:0\le n\le N\}$ and
$\{\tP_n:0\le n\le N\}$ must be isospectral. Even more, since the matrices $\sf H$ and
$\widetilde{\sf H}$ of the single-particle Hamiltonians of both chains are related by
$\widetilde{\sf H}=O^T\sf H O$, where $O$ is a real orthogonal matrix, it follows from
Eq.~\eqref{HC} that the Hamiltonian $\widetilde H$ is mapped into $H$ by the unitary
transformation
\begin{equation}\label{tcOc}
  \widetilde c_i=\sum_{j=0}^{N-1}O_{ij}c_j
\end{equation}
between their respective sets of fermionic operators. We can thus equally well use the sets
$(P,Q,R)$ or $(\tP,\tQ,\tR)$ to construct the QES model~\eqref{V} and its corresponding XX chain,
up to the isomorphism determined by Eq.~\eqref{tcOc}. In view of this residual symmetry in the
description of a QES model and its associated chain or free fermion system, we can apply to a
quartic polynomial $P(z)$ with real coefficients satisfying conditions~\eqref{condP}, i.e.,
vanishing at zero and infinity on the extended real line, a \emph{real} projective
transformation~\eqref{tPQR} taking it to a simpler (canonical) form still satisfying the latter
conditions. It is possible in this way to reduce any quartic polynomial $P$ satisfying
conditions~\eqref{condP} to the six inequivalent types of canonical forms listed in
Table~\ref{tab.canP} (see the appendix for a detailed proof).
\begin{table}
  \centering
\begin{tabular}{ll}
\noalign{\smallskip}\hline\noalign{\smallskip}
  1. & $\nu z(1+z)$\\
  2. & $\nu z(1-z)$\\
  3. & $\nu z^2$\\
  4. & $z$\\
  5. & $\nu z(1+z)(a+z)\,,\quad\text{with}\en 0<a<1$\\
  6. & $\nu z(z^2+2a z+1)\,,\quad\text{with}\en -1<a<1$\\
\noalign{\smallskip}\hline
\end{tabular}
\caption{Inequivalent canonical forms of a quartic polynomial vanishing at the origin and infinity
  under projective transformations~\eqref{tPQR} (in all cases, $\nu$ is a positive constant).}
\label{tab.canP} 
\end{table}

\subsection{Canonical forms of XX chains constructed from QES potentials}

We shall next construct the QES potentials $V(x)$ and the XX chains determined by the six
canonical forms listed in Table~\ref{tab.canP}. Note, in this respect, that the parameter $\nu>0$
appearing in most of these canonical forms can be chosen at will by rescaling the $x$ coordinate.
By Eq.~\eqref{QESchain}, multiplying $P$ and $Q$ by a constant factor~$\la>0$ merely rescales the
parameters $J_n$ and $B_n$ in the associated XX chain or free fermion system by the same factor.
Hence, without loss of generality, we shall fix the parameter $\nu$ appropriately in each case to
simplify the expression for the potential $V(x)$.

\bigskip
\noindent 1.\en $P(z)=4z(1+z)$
\medskip

\noindent
The change of variables relating the variable $z$ to the physical coordinate $x$ is given in this
case by
\[
  z=\sinh^2x\,,
\]
up to an irrelevant translation $x\mapsto x-x_0$. Setting
\[
  Q(z)=-8\al z^2+4(-2\al+\be+\ga+N-1)z+2(2\ga+N-1)\,,
\]
where $\al,\be,\ga$ are real parameters, we obtain the following formula for the pseudo-gauge
factor $\mu$ in Eq.~\eqref{xmu}:
\[
  \mu=\e^{-\frac\al2\cosh 2x}(\cosh x)^\be(\sinh x)^\ga\,,
\]
up to an inessential multiplicative constant. By Eq.~\eqref{V}, the potential $V(x)$ is given by
\begin{eqnarray*}
  V(x)=\frac{\al^2}2\,\cosh 4x&-2\al(\be+\ga+2N-1)\cosh
                                             2x\\
  &-\be(\be-1)\sech^2x+\ga(\ga-1)\csch^2x+V_0,
\end{eqnarray*}
with
\[
  V_0=-h_*-\frac{\al^2}2+2\al(\be-\ga)+(\be+\ga-1)(\be+\ga+2N-1)+N^2\,.
\]
From Eqs.~\eqref{PnQES}-\eqref{algeig} it follows after a straightforward calculation that the
(unnormalized) algebraic eigenfunctions can be expressed in terms of the monic polynomials
$P_n(E)$ by the formula
\[
  \fl
  \psi_k(x)=\e^{-\frac\al2\cosh 2x}(\cosh x)^\be(\sinh
  x)^\ga\,\sum_{n=0}^{N-1}\frac{P_n(E_k)}{(\ga+\frac12)_n}\,\frac{\big(-\frac14\sinh^2x\big)^n}{n!}\,,
  \qquad 0\le k\le N-1\,,
\]
where $(a)_n:=a(a+1)\cdots(a+n-1)$ denotes the shifted factorial. Thus the square integrability at
infinity of the algebraic eigenfunctions requires that $\al\ge0$. On the other hand, the last
nonconstant term in the potential is singular at the origin unless $\ga=0$ or $\ga=1$. For these
values of $\ga$ the algebraic eigenfunctions are respectively even or odd functions of the
variable $x$. Moreover, when $\ga\ne0,1$ the square integrability at the origin of the algebraic
eigenfunctions is guaranteed provided that $\ga>-1/2$, but the stronger condition $\ga>1/2$ is
required so ensure that the Hamiltonian is essentially selfadjoint. For $1/2<\ga<1$ the potential
is unbounded below near $0$, while for $\ga>1$ it tends to $+\infty$ as $x^{-2}$ when $x\to0$. The
impenetrable nature of the potential barrier near $0$ in the latter case implies that the particle
is effectively confined either to the positive half-line $(0,\infty)$ or to its negative
$(-\infty,0)$.

Using Eq.~\eqref{QESchain} we obtain the following formula for the parameters
of the associated XX chain or free fermion system:
\begin{eqnarray*}
  J_n&=4\sqrt{\al(n+1)(N-n-1)(2\ga+2n+1)}\,,\\
  B_n&=-h_*-4\big(n-\tfrac{N-1}2\big)\big(n+\tfrac{N-1}2-2\al+\be+\ga\big)\,.
\end{eqnarray*}
Taking into account that $\al\ge0$, condition~\eqref{condan} (which is tantamount to requiring
that $J_n$ be real and nonvanishing for $0\le n\le N-2$) is in this case $\al>0$ and
$\ga>-\frac12$. Condition~\eqref{hPncond} is automatically satisfied, since
$A_{n+1}=-2(2\ga+2n+1)$. We thus see that the conditions
\begin{equation}\label{alga1}
  \al>0,\qquad \ga=0\en\text{or}\en\ga>\frac12\,,
\end{equation}
guarantee both the regularity of the algebraic eigenfunctions and the existence of the associated
inhomogeneous XX chain or free fermion system.

\bigskip
\noindent 2.\en $P(z)=4z(1-z)$
\medskip

\noindent
This is the trigonometric version of the previous case. More precisely, setting
\[
  Q(z)=-8\al z^2 + 4(2\al-\be-\ga-N+ 1) z+2(2\ga+N -1)
\]
the change of variables and pseudo-gauge factor are given by
\[
  z=\sin^2x\,,\qquad\mu(x)=\e^{-\frac\al2\cos2x}(\cos x)^\be(\sin x)^\ga\,.
\]
The potential $V(x)$ in this case reads
\begin{eqnarray*}
  V(x)=-\frac{\al^2}2\,\cos 4x&+2\al(\be+\ga+2N-1)\cos
                                2x\\
  &+\be(\be-1)\sec^2x+\ga(\ga-1)\csc^2x+V_0,
\end{eqnarray*}
with
\[
  V_0=-h_*+\frac{\al^2}2-2\al(\be-\ga)-(\be+\ga-1)(\be+\ga+2N-1)-N^2\,.
\]
The corresponding algebraic eigenfunctions are obtained from Eqs.~\eqref{PnQES}-\eqref{algeig},
with the result
\[
  \psi_k(x)=\e^{-\frac\al2\cos2x}(\cos x)^\be(\sin
  x)^\ga\sum_{n=0}^{N-1}\frac{P_n(E_k)}{(\ga+\frac12)_n}
  \frac{\big(-\frac14\sin^2x\big)^n}{n!}\,.
\]
Although this potential is obtained from the previous one applying the Wick rotation
$V(x)\mapsto-V(\iu x)$, the physical natures of these potentials are quite different. Indeed, if
either $\be$ or $\ga$ are $0$ or $1$, $V(x)$ is a nonsingular $\pi$-periodic potential and the
algebraic eigenfunctions are either periodic (if $\be+\ga$ is even) or antiperiodic (if $\be+\ga$
is odd) functions, and thus belong to the edges of the band spectrum. The functions $\sec^2x$ and
$\csc^2x$ behave near their respective singularities at $x_k=(2k+1)\pi/2$ and $x_k=k\pi$ (with
$k\in\ZZ$) as $(x-x_k)^{-2}$. For $\be\notin\{0,1\}$ the square integrability of the algebraic
eigenfunctions near the singularities of $\sec^2x$ will be guaranteed provided that $\be>-1/2$,
and similarly $\ga>-1/2$ when $\ga\notin\{0,1\}$. As in the previous case, however, to ensure that
$H$ is essentially self-adjoint we need the stronger conditions $\be>1/2$ or $\ga>1/2$,
respectively. Moreover, if $\be,\ga>1$ the potential confines the particle inside the interval
$(0,\pi/2)$ (modulo $\pi/2$), and has a purely discrete spectrum.

From Eq.~\eqref{QESchain} we readily obtain the coefficients of the XX chain
or free fermion system in this case:
\begin{eqnarray*}
  J_n&=4\sqrt{\al(n+1)(N-n-1)(2\ga+2n+1)}\,,\\
  B_n&=-h_*+4\big(n-\tfrac{N-1}2\big)\big(n+\tfrac{N-1}2-2\al+\be+\ga\big)\,.
\end{eqnarray*}
Note that $J_n$ is the same as in the previous case, while $B_n+h_*$ differs from its counterpart
for Case 1 one only in its sign. Finally, taking into account the restrictions on the parameters
$\be$ and $\ga$ coming from the regularity of the algebraic eigenfunctions, the conditions
ensuring that $J_n$ is real and nonzero are in this case given by
\[
  \al>0\,,\quad\bigg(\be=0\en\text{or}\en\be>\frac12\bigg)\,,\quad
  \bigg(\ga=0\en\text{or}\en\ga>\frac12\bigg)\,.
\]
As in the previous case, these conditions also guarantee that $A_{n+1}=-2(2\ga+2n+1)\ne0$ for all
$n\ge0$.

\bigskip
\noindent 3.\en $P(z)=z^2$
\medskip

\noindent Parametrizing $Q(z)$ as
\[
  Q(z)=-2\al z^2+\be z+2\ga
\]
we have
\[
  z=\e^x\,,\qquad \mu=\exp\Bigl(-\al\e^x-\ga\e^{-x}+\tfrac12(\be-N+1)x\Bigr)
\]
and
\[
  V(x)=\al^2\e^{2x}+\ga^2\e^{-2x}-\al(\be+N)\e^x+\ga(\be-N)\e^{-x}+V_0\,,
\]
with
\[
  V_0=-h_*-2\al\ga+\frac14\,\be^2\,.
\]
Note that in this case the square integrability of the eigenfunctions, given by
\begin{equation}\label{alg2}
  \fl
  \psi_k(x)=\exp\Bigl(-\al\e^x-\ga\e^{-x}+\tfrac12(\be-N+1)x\Bigr)
  \sum_{n=0}^{N-1}\frac{P_n(E_k)}{n!}\,\bigg(-\frac{\e^x}{2\ga}\bigg)^n,
  \en0\le k\le N-1,
\end{equation}
requires that $\al,\ga\ge0$. The coefficients of the inhomogeneous XX chain associated with this
model are found to be
\begin{equation}\label{Jnan2}
  \fl
  J_n=2\sqrt{\al\ga(n+1)(N-n-1)}\,,\quad B_n=-h_*-\frac14\,(2n-N+1)(2\be+2n-N+1);
\end{equation}
note, in particular, that $J_n$ will be real and nonvanishing provided that
\[
  \al>0\,,\qquad\ga>0\,,
\]
and hence $A_{n+1}=-2\ga\ne0$. The latter conditions also guarantee the square integrability of
the algebraic eigenfunctions. Note also that the hopping amplitude~\eqref{Jnan2} coincides with
that of the chain derived in~\cite{CNV19} from the classical Krawtchouk polynomials
$K_n(x;p,N-1):={}_2F_1\Bigl( {-n,-x\atop -N+1};1/p\Bigr)$ \cite{Kr29}, provided that
$p=(1\pm\sqrt{1-16\al^2\ga^2}\,)/2$. It can be shown, however, that the coefficients $B_n$ in both
chains differ regardless of the value of the remaining parameter $\be$. Thus the chain with
coefficients~\eqref{Jnan2} appears to be new.

\bigskip
\noindent 4.\en $P(z)=4z$
\medskip

\noindent
Although the coefficient multiplying $z$ in this case can be made equal to one by a suitable
dilation, we have taken without loss of generality $\nu=4$ for later convenience. Writing
\[
  Q(z)=-4\al z^2+4\be z+2(2\ga+N-1)
\]
we obtain the following formulas for the change of variable and the pseudo-gauge factor:
\[
  z=x^2\,,\qquad \mu=x^\ga\e^{-\frac\al4x^4+\frac\be2x^2}\,.
\]
Thus $\al\ge0$ is necessary to ensure square integrability at infinity of the algebraic
eigenfunctions. The potential in this case is given by
\[
  V(x)=\al^2 x^6-2\al\be x^4+(\be^2-2\al\ga+\al-4\al N)x^2+\frac{\ga(\ga-1)}{x^2}+V_0\,,
\]
with
\[
  V_0=-h_*+\be(2\ga+2N-1)\,.
\]
As in Case 1, if $\ga\ne0,1$ the self-adjointness of $H$ requires that $\ga>1/2$, while for
$\ga>1$ the potential effectively confines the particle either to the positive or the negative
half-line. The algebraic eigenfunctions are now given by
\[
  \psi_k(x)=x^\ga\e^{-\frac\al4x^4+\frac\be2x^2}\sum_{n=0}^{N-1}\frac{P_n(E_k)}{n!(\ga+\frac12)_n}\,
  \bigg(-\frac{x^2}4\bigg)^n\,,\qquad 0\le k\le N-1\,.
\]
Note that when $\ga=0$ the algebraic eigenfunctions are all even, whereas for $\ga=1$ they are
odd.

From Eq.~\eqref{QESchain} it follows that in this case the parameters of the
associated inhomogeneous XX chain or free fermion system are given by
\[
  \fl
  J_n=4\sqrt{\al(n+1)(N-n-1)(\ga+n+\tfrac12)}\,,\qquad
  B_n=-h_*-2\be(2n-N+1)\,.
\]
Since $\al\ge0$, we see that in this case condition~\eqref{condan} holds ---i.e., $J_n$ is real
and nonzero--- provided that $\al>0$ and $\ga>-\frac12$. Hence, as in Case 1, the
conditions~\eqref{alga1} guarantee both the regularity of the algebraic eigenfunctions and the
existence of the associated inhomogeneous XX chain or free fermion system. Finally,
$A_{n+1}=-2(2\ga+2n+1)\ne0$ for all $n\ge0$ on account of~\eqref{alga1}.

\bigskip
\noindent 5.\en $P(z)=4z(1+z)(a+z)$, with $0<a<1$.
\medskip

\noindent
Writing
\[
  \fl
  Q(z)=-2(2\al+N-1)z^2-4\big[\al+k^2(N-\be)-k'^2\ga\big]\,z+2k'^2(2\ga+N-1),
\]
the change of variables and pseudo-gauge factor can be taken as
\[
  z=\frac{\cn^2x}{\sn^2x}\,,\qquad \mu=(\cn x)^\ga(\sn x)^{\al+2N-2}(\dn x)^{\la}\,,
\]
where
\begin{equation}\label{lambda}
  \la:=-\al+\be-\ga-2N+1
\end{equation}
and $\cn x\equiv\cn(x;k)$, $\sn x\equiv\sn(x;k)$, and $\dn x\equiv\dn(x;k)$ are the standard
Jacobi elliptic functions with (square) modulus $k^2:=1-a\in(0,1)$. A long but straightforward
calculation yields the following formula for the potential $V(x)$ of the corresponding QES model:
\begin{equation}\label{Vell}
  \fl
  V(x)=\frac{\al(\al-1)}{\sn^2x}+k'^2\,\frac{\ga(\ga-1)}{\cn^2x}
  -k'^2\,\frac{\la(\la-1)}{\dn^2x}+k^2\be(\be-1)\sn^2x+V_0\,,
\end{equation}
with
\[
  \fl V_0=-h_*+\be+k'^2\big[\be^2+(2N-2\be+1)\ga\big]+(2-k^2)N(N-2\be) +\al(2N-2\be+1)\,.
\]
The algebraic eigenfunctions are given by
\begin{eqnarray*}
  \psi_k(x)
  &=(\cn x)^\ga(\sn x)^{\al+2N-2}(\dn
    x)^{\la}\sum_{n=0}^{N-1}\frac{P_n(E)}{n!(\ga+\frac12)_n}\,
    \biggl(-\frac{\cn^2x}{4k'^2\sn^2x}
    \biggr)^n\\
  \fl
  &=(\dn
    x)^{\la}\sum_{n=0}^{N-1}\biggl(-\frac{1}{4k'^2}\biggr)^n
    \frac{P_n(E)}{n!(\ga+\frac12)_n}\,
    (\cn x)^{\ga+2n}(\sn x)^{\al+2(N-1-n)}\,.
\end{eqnarray*}
The potential is regular everywhere if and only if $\al$ and $\ga$ are both either $0$ or $1$. In
this case $V$ is $2K$-periodic, where
\[
  K\equiv K(k):=\int_0^{\pi/2}\frac{\diff\th}{1-k^2\sin^2\th}
\]
is the complete elliptic integral of the first kind. The algebraic eigenfunctions are not
square-integrable, but belong to the continuous spectrum (in fact, to the boundaries of the band
spectrum). On the other hand, when $\al\ne0,1$ the potential $V(x)$ diverges at the real zeros
$2mK$ (with $m\in\ZZ$) of $\sn$ as $\al(\al-1)(x-2mK)^{-2}$. In this case the square integrability
of the algebraic eigenfunctions at the singularities of $\sn^{-2}x$ is guaranteed provided that
$\al>-1/2$, while the stronger condition $\al>1/2$ is needed to ensure the self-adjointness of
$H$. Likewise, if $\ga\ne0,1$ then $V$ diverges as $\ga(\ga-1)(x-(2m+1)K)$ (with $m\in\ZZ$) at the
real zeros $(2m+1)K$ of $\cn$. Hence in this case square integrability of the algebraic
eigenfunctions requires that $\ga>-1/2$, while $\ga>1/2$ is needed for $H$ to be essentially
self-adjoint. Finally, if both $\al>1$ and $\ga>1$ then the particle is confined inside the
interval $(0,K)$ (modulo $K$), and the spectrum of $H$ is purely discrete.

The XX spin chain or free fermion system associated to the elliptic QES potential~\eqref{Vell} has
parameters
\begin{eqnarray*}
  \fl
  J_n&=4k'\sqrt{(n+1)(N-n-1)(n+\ga+\tfrac12)(N-n+\al-\tfrac32)},
  \\
  \fl
  B_n&=-h_*+2(2n-N+1)\big[\al-\ga-2n+N-1+k^2\big(\ga-\be+n+\tfrac12\,(N+1)\big)\big]\,.
\end{eqnarray*}
Taking into account the regularity conditions on the eigenfunctions discussed above,
Eq.~\eqref{condan} will hold provided that $\al,\ga>-\frac12$, in which case
$A_{n+1}=-2k'^2(2\ga+2n+1)\ne0$ for all $n\ge0$. Thus the conditions
\begin{equation}\label{alga}
  \bigg(\al=0\en\text{or}\en\al>\frac12\bigg)\,,\quad \bigg(\ga=0\en\text{or}\en\ga>\frac12\bigg)
\end{equation}
guarantee both the regularity of the algebraic eigenfunctions and the existence of the associated
XX chain or free fermion system.

\bigskip
\noindent 6.\en $P(z)=z\big(z^2+2az+1\big)$, with $-1<a<1$. \medskip

\noindent Setting
\[
  Q(z)=-\frac12\,(2\al+N-1)z^2+\big[\be-(1-2k^2)(\al-\ga)\big]z+\ga+\frac12\,(N-1)\,,
\]
the change of variables and pseudo-gauge factor are
\begin{eqnarray*}
  \fl
  z&=\frac{1+\cn x}{1-\cn x}\,,\\
  \fl
  \mu&=(1+\cn x)^{\ga/2}(1-\cn
       x)^{\frac\al2+N-1}(\dn x)^{-\frac12(\al+\ga)-N+1}
       \exp\Bigl(\tfrac{\be}{4kk'}\arctan\big(\tfrac{\cn
       x+\dn^2x}{kk'\sn^2x})\Bigr)\,,
\end{eqnarray*}
where now the square modulus of the elliptic functions is $k^2:=(1-a)/2\in(0,1)$. The
potential~\eqref{V} is given by
\begin{equation}\label{Vell6}
  V(x)=\frac{A+B\cn x}{\sn^2x}+\frac{C+D\cn x}{\dn^2x}+V_0\,,
\end{equation}
with
\begin{eqnarray*}
  \fl
  A&=\frac12\big[\al(\al-1)+\ga(\ga-1)\big]\,,
  &B=\frac12\,(\al-\ga)(\al+\ga-1)\,,\\
  \fl
  C&=\frac{\be^2}{16k^2}-\frac{k'^2}4(\al+\ga+2N-2)(\al+\ga+2N)\,,\qquad
  &D=-\frac\be{4k^2}(\al+\ga+2N-1)\,,\\
  \fl
  V_0&={-h_*-\frac{\be^2}{16k^2}-\frac{k^2}4\,(\al-\ga)^2+\frac14\,\big[\be(\ga-\al)
       +2N(\al+\ga+N-1)\big]\,.}\hidewidth
\end{eqnarray*}
The algebraic eigenfunctions are in this case
\begin{eqnarray*}
  \psi_k(x)&=\mu\sum_{n=0}^{N-1}\frac{P_n(E_k)}{n!(\ga+\frac12)_n}\,\bigg(-\frac{1+\cn x}{1-\cn
  x}\bigg)^n\\
  &=(\dn x)^{-\frac12(\al+\ga)-N+1}
       \exp\Bigl(\tfrac{\be}{4kk'}\arctan\big(\tfrac{\cn
  x+\dn^2x}{kk'\sn^2x})\Bigr)\\&\hphantom{ = (\dn}\times
  \sum_{n=0}^{N-1}(-1)^n\frac{P_n(E_k)}{n!(\ga+\frac12)_n}\,(1+\cn x)^{\frac\ga2+n}
  (1-\cn x)^{\frac\al2+N-n-1}\,.
\end{eqnarray*}
From the identity
\[
  \frac{A+B\cn x}{\sn^2x}=\frac{\al(\al-1)}{2(1-\cn x)}+\frac{\ga(\ga-1)}{2(1-\cn x)}
\]
it follows that if $\al\ne0,1$ (resp.~$\ga\ne0,1$) the potential is singular at the real zeros
$4mK$ of $1-\cn x$ (resp.~the real zeros $2(2m+1)K$ of $1+\cn x$), where $m$ is an integer. Again,
in the first case the regularity conditions are $\al>-1/2$ for square integrability of the
algebraic eigenfunctions and $\al>1/2$ for the Hamiltonian to be essentially selfadjoint, and
similarly for $\ga\ne0,1$. Moreover, if $\al,\ga>1$ the potential confines the particle inside the
finite interval~$(0,2K)$ modulo $2K$.

The associated XX chain or free fermion model coefficients are
\begin{equation}\label{JBn6}
  \eqalign{J_n&=\sqrt{(n+1)(N-n-1)(n+\ga+\tfrac12)(N-n+\al-\tfrac32)}\,,\cr
  B_n&=-h_*+\frac12(2n-N+1)\big[-\be+(1-2k^2)(\al-\ga-2n+N-1)\big]\,.}
\end{equation}
Like in the previous case, Eq.~\eqref{condan} is satisfied provided that both $\al$ and $\ga$ are
greater than $-1/2$, which also implies that $A_{n+1}=-(\ga+n+1/2)\ne0$ for all $n\ge0$. Hence the
conditions guaranteeing the regularity of the algebraic eigenfunctions and the existence of the
associated XX chain or free fermion system are again given by Eq.~\eqref{alga}. Note, finally,
that the hopping amplitude~\eqref{JBn6} coincides with that of the chain constructed
in~\cite{CNV19} from the dual Hahn polynomials
\[
  R_n(x(x+\al+\ga);\ga-1/2,\al-1/2,N-1):={}_3F_2\biggl( {-n,x+\al+\ga,-x\atop
    \ga+1/2,-N+1};1\biggr)
\]
(cf.~\cite{ha49}). However, as in Case 3, the coefficients $B_n$ in both chains differ for all
values of the parameters $\al$, $\be$, $\ga$ and $k$. Hence the chain with
coefficients~\eqref{JBn6} appears to be new.

\subsection{The Lamé chains}

The Lamé (finite gap) potential is defined by
\begin{equation}\label{Lame}
  V(x)=k^2l(l+1)\sn^2x\,,
\end{equation}
where $k\in(0,1)$ is the modulus of the elliptic sine and $l\ge-1/2$ is a real
parameter~\cite{Ar64}. This potential has important applications in many areas of mathematics and
physics, such as potential theory (indeed, it arises by separation of variables in Laplace's
equation in ellipsoidal coordinates), the theory of crystals~\cite{AGI83}, field
theory~\cite{BB93} and inflationary cosmology~\cite{GKLS97,FGLR00}. Since the
potential~\eqref{Lame} is smooth and $2K(k)$-periodic, the corresponding
Hamiltonian~\eqref{QESmodel} has a purely continuous (band) energy spectrum. It is well known that
when $l$ is a nonnegative integer the spectrum has exactly $l$ gaps, and the $2l+1$ eigenfunctions
belonging to the boundaries of the allowed energy bands are homogeneous polynomials in the
Jacobian elliptic functions $\sn$, $\cn$ and $\dn$ (the so called Lamé polynomials). On the other
hand, when $l$ is a positive half-integer the Lamé potential admits (for characteristic values of
the energy $E$) two linearly independent non-meromorphic eigenfunctions with period $8K(k)$
expressible in closed form in terms of $\sn$, $\cn$ and $\dn$~\cite{FGR00}.

The Lamé potential can be obtained as a particular case of the elliptic QES models in Cases 5 and
6 in the previous section. Indeed, consider to begin with Case 5. It is clear that the
potential~\eqref{Vell} in this case reduces to the Lamé potential~\eqref{Lame} (up to a constant,
which can be taken equal to zero by choosing $h_*$ appropriately) provided that the parameters
$\al$, $\ga$ and $\la$ take independently the values $0$ or $1$. Taking into account the
definition~\eqref{lambda} of~$\la$, this means that
\[
 \al=\vep_1,\qquad \ga=\vep_2\,,\qquad\la=\vep_3\,,\qquad\be=2N-1+\vep_1+\vep_2+\vep_3\,,
\]
with $\vep_i\in\{0,1\}$ independently. The parameter $l$ is then given by
\[
  l=\be-1=2(N-1)+\vep_1+\vep_2+\vep_3\,,
\]
the alternative solution $l=-\be$ being unacceptable on account of the condition $l\ge-1/2$. Thus
in this case the parameter $l$ is an integer. The corresponding XX chain has parameters
\begin{equation}
  \label{JBell5}
  \fl
  \eqalign{J_n&=4k'\sqrt{(n+1)(N-n-1)(n+\tfrac12+\vep_2)(N-n-\tfrac32+\vep_1)}\,,\cr
    B_n&=-h_*-2(2n-N+1)\big[2n-N+1-\vep_1+\vep_2+k^2\big(n-\tfrac32(N-1)-\vep_1-\vep_3\big)\big].}
  \end{equation}
Note that $J_n$ is symmetric under $n+1\mapsto N-n-1$ provided that $\vep_1=\vep_2$.

Consider next the potential~\eqref{Vell6} in Case 6. In fact, since $V(x)$ is defined up to a translation
in $x$, it is equivalent but more convenient for our purposes to consider the potential $V(x+K)$.
From the well-known identities
\begin{eqnarray*}
  \cn(x+K)=-k'\,\frac{\sn x}{\dn x}\,,\quad
  \sn^{-2}(x+K)=\frac{\dn^2x}{\cn^2x}\,,\\
  \dn^{-2}(x+K)=\frac{\dn^2x}{k'^2}=\frac1{k'^2}(1-k^2\sn^2x)\,,
\end{eqnarray*}
it follows that $V(x+K)$ reduces to the Lamé potential~\eqref{Lame} (up to a constant) provided
that $A=B=D=0$. The general solution of the latter equations is
\[
  \al=\vep_1\,,\qquad \be=0\,,\qquad \ga=\vep_2\,,
\]
where again $\vep_1,\vep_2\in\{0,1\}$ independently. The parameter $l$ is given by
\[
  l=N-1+\frac12(\vep_1+\vep_2)\,,
\]
and is thus an integer (if $\vep_1=\vep_2$) or a half-integer (if $\vep_1=1-\vep_2$). The
coefficients of the corresponding XX chain or free fermion systems are in this case
\begin{equation}\label{JBell6}
    \eqalign{J_n&=\sqrt{(n+1)(N-n-1)(n+\tfrac12+\vep_2)(N-n-\tfrac32+\vep_1)}\,,\cr
    B_n&=-h_*-\frac12(1-2k^2)(2n-N+1)(2n-N+1-\vep_1+\vep_2)\,.}
  \end{equation}
In particular, both the hopping amplitude $J_n$ and the magnetic field $B_n$ are symmetric for
$\vep_1=\vep_2$. Moreover, the parameter $B_n$ clearly vanishes when $k^2=1/2$ if we take $h_*=0$.

\section{Entanglement entropy of a Lamé chain}\label{sec.ent}

The bipartite entanglement entropy of a quantum system consisting of two subsystems $A,B$ in a
(pure or mixed) state with density matrix~$\rho$ is defined as
\[
  S_A:=s[\rho_A]\,,
\]
where $\rho_A:=\tr_B\rho$ is the reduced density matrix of subsystem $A$ and $s$ is any entropy
functional. In fact, when $\rho$ is a pure state (as we shall assume in the sequel) Schmidt's
decomposition theorem~\cite{NC10} implies that~$s[\rho_A]=s[\rho_B]$ , so that $S_A=S_B$. A common
choice of $s$ is the Rényi entropy
\[
  s_\al[\rho_A]=\frac1{1-\al}\log\tr(\rho_A^\al)\,,
\]
where $\al>0$ is a real parameter. Its limit as $\al\to1$ is the von Neumann (or Shannon) entropy
\[
  s_1[\rho_A]:=\lim_{\al\to1}s_\al[\rho_A]=-\tr(\rho_A\log\rho_A)\,.
\]
The exact evaluation of $S_A$ is in general impossible and its numerical computation is also
prohibitive even for relatively small systems, since it entails the determination of the
eigenvalues of the matrix~$\rho_A$. For instance, if $A$ is a set of $L$ consecutive sites of a
chain of spins $1/2$ the size of $\rho_A$ is $2^L$, which grows exponentially with $L$.
Remarkably, however, for systems like the chain~\eqref{Hchain} or the equivalent free fermion
system~\eqref{Hffs}, whose energy eigenstates are Slater determinants, there is a well-known
algorithm for computing $S_A$ based on the diagonalization of an $L\times L$
matrix~\cite{Pe03,VLRK03}. More precisely, suppose that the free fermion system~\eqref{Hffs}
is in the energy eigenstate
\[
  \ket{M}:=\tilde c^\dagger_0\tilde c^\dagger_1\cdots\tilde c^\dagger_{M-1}\ket0
\]
in which the lowest $M$ single-body energies $E_k$ are excited, and let $A$ be the subsystem
consisting of the first $L$ fermions $0,\dots,L-1$. We define the correlation matrix
$C_A=(C_{ij})_{0\le i,j\le L-1}$ by setting
\[
  C_{ij}=\bra M \dc_ic_j\ket M\,.
\]
The bipartite entanglement entropy $S_A$ can then be computed through the formula
\begin{equation}
  \label{SA}
  S_A=\sum_{i=0}^{L-1}s^{(2)}(\nu_i)\,,
\end{equation}
where $\nu_0,\dots,\nu_{L-1}$ are the eigenvalues of $C_A$ and $s^{(2)}(x):=s[\diag(x,1-x)]$ is
the binary entropy associated with $s$. It is indeed easy to see that $C_A$ is Hermitian, and that
both $C_A$ and $1-C_A$ are positive semi-definite, so that $0\le\nu_l\le1$. For instance, for the
Rényi entropy $s_\al$ we have
\[
  s_\al^{(2)}(x)=\frac1{1-\al}\log\bigl(x^\al+(1-x)^\al\bigr)\,,
\]
while for the von Neumann one, which from now on we shall simply denote by $s$,
\[
  s^{(2)}(x)=-x\log x-(1-x)\log(1-x)
\]
(with $0\log0:=0$).

The correlation matrix $C_A$ can be easily expressed in terms of the OPS $\{P_n\}_{n=0}^N$
associated with the system~\eqref{Hffs}. Indeed, it suffices to note that
\[
  \bra M \tc^\dagger_n\tc_m\ket M=\bra M \tc^\dagger_n\tc_n\ket M\de_{nm}=\de_{nm}\chi_\cM(n)\,,
\]
where $\chi_\cM$ is the characteristic function of the set $\cM=\{0,\dots,M-1\}$ (i.e.,
$\chi_\cM(n)=1$ for $0\le n\le M-1$ and $\chi_\cM(n)=0$ for $n\ge M$). Expressing the fermionic
operators $\dc_i$ and $c_j$ in terms of their counterparts $\tc^\dagger_n$, $\tc_m$ using the
inverse of Eq.~\eqref{mommodes}, i.e.,
\[
  c_k=\sum_{n=0}^{N-1}\phi_k(E_n)\tc_n
\]
(where we have taken into account that $\Phi=\{\phi_k(E_n)\}_{0\le k,n\le N-1}$ is orthogonal) we
easily arrive at the formula
\[
  C_{ij}=\sum_{n=0}^{M-1}\phi_i(E_n)\phi_j(E_n)
\]
or, using Eqs.~\eqref{gan}, \eqref{wkexp}, and~\eqref{phiP},
\begin{equation}
  \label{Cij}
  C_{ij}=\sum_{n=0}^{M-1}\frac{w_n}{\ga_n}P_i(E_n)P_j(E_n)=
  \sum_{n=0}^{M-1}\prod_{k=n+1}^{N-1}a_k\cdot\frac{P_i(E_n)P_j(E_n)}{P_{N-1}(E_n)P_N'(E_n)}\,.
\end{equation}

Equation~\eqref{Cij} can be used to efficiently compute the correlation matrix $C_A$, and hence
the entanglement entropy $S_A$ through equation~\eqref{SA} after diagonalizing $C_A$, for all the
spin chains constructed in the previous section. As an example, we shall next use the latter
formula to study the entanglement entropy of one of the Lamé chains~\eqref{JBell5}
and~\eqref{JBell6}. For simplicity, we have chosen the symmetric version of~\eqref{JBell6} with
$\vep_1=\vep_2=0$, for which $l=N-1$, and have also set $h_*=0$, so that
\begin{equation}
  \label{Lamechain}
  \eqalign{J_n&=\sqrt{(n+1)(N-n-1)(N-n-\tfrac32)(n+\tfrac12)}\,,\cr
  B_n&=-\frac12(1-2k^2)(2n-N+1)^2\,.}
\end{equation}
For $k^2=1/2$, this chain has $B_n=0$ for all $n$. This makes it possible to obtain an asymptotic
formula for its Rényi entanglement in the half-filling regime $M=\lfloor N/2\rfloor$ (where
$\lfloor\cdot\rfloor$ denotes the integer part) using its connection with an appropriate conformal
field theory. (Note that for any nontrivial choice of $M\in\{1,\dots,N-2\}$ we can regard $\ket M$
as the system's ground state by choosing $E_{M}<h_*<E_{M+1}$.)

More precisely, consider in general a spin chain of the form~\eqref{Hchain} or its equivalent free
fermion system~\eqref{Hffs} with $B_n=0$ for all $n$, which we shall rewrite in the more symmetric
fashion
\[
  H=\sum_{m=-N/2+1}^{N/2-1}J_{m+N/2-1}\big(d^\dagger_{m}d_{m+1}+d^\dagger_{m+1}d_{m}\big)\,,\qquad
  d_m:=c_{m+N/2-1}\,.
\]
Here $m=-N/2+1,-N/2,\dots,N/2-1$ can be integer or half-integer according to whether $N$ is even
or odd. We shall first derive the continuum limit of the latter model by introducing a site
spacing $a$, setting $x=ma$ and letting $a\to0$ and $N\to\infty$ in such a way that $a(N-1)/2$
tends to a finite limit~$\ell$, equal to the chain's half-length. Following
Refs.~\cite{Kat12,RRS14,RRS15,RDRCS17,TRS18}, we expand the fermionic operators $d_m$ in slow
modes $\psi_{\mathrm L,\mathrm R}(x)$ around the Fermi points $\pm\kF$, where~$\kF$ is the Fermi
momentum at half filling, as
\[
  d_{m}\simeq\sqrt a\Bigl(\e^{\iu\kF x}\psiL(x)+\e^{-\iu\kF x}\psiR(x)\Bigr)\,.
\]
At half filling we have $\kF=\pi/(2a)$, so that
\begin{eqnarray*}
  \frac{d_{m+1}}{\sqrt a}&\simeq\iu\Bigl(\e^{\iu\kF x}\psiL(x+a)-\e^{-\iu\kF x}\psiR(x+a)\Bigr)\\
  &\simeq\iu\Bigl(\e^{\iu\kF x}\psiL(x)-\e^{-\iu\kF x}\psiR(x)\Bigr)
  +\iu a\Bigl(\e^{\iu\kF x}\pd_x\psiL(x)-\e^{-\iu\kF x}\pd_x\psiR(x)\Bigr)\,.
\end{eqnarray*}
If the fields $\psi_{\mathrm L,\mathrm R}(x)$ are slowly varying, cross terms in
$d^\dagger_{m}d_{m+1}+d^\dagger_{m+1}d_m$ like $\iu a\e^{2\iu\kF x}\psiR^\dagger(x)\psiL(x)$
vanish when summed over $m$. Calling
\[
  J_{m+N/2-1}=J_{(x+\ell)/a-\frac12}:=J(x;a)\,,
\]
and taking into account that $x=ma\in[-\ell+a/2,\ell-a/2$], we are left with
\begin{equation}\label{HJxa}
  H\simeq\iu a\int_{-\ell}^\ell J(x;a)\big[\psiL^\dagger(x)
  \overset{\leftrightarrow}{\pd_x}\psiL(x)
  -\psiR^\dagger(x)\overset{\leftrightarrow}{\pd_x}\psiR(x)\big]\diff x\,.
\end{equation}
For the Lamé chain~\eqref{Lamechain} with $k^2=1/2$ we have
\begin{eqnarray}
  \fl
  J(x;a)&=\frac{\ell^2}{a^2}\sqrt{\bigg(1-\frac{x^2}{\ell^2}\bigg)\bigg(\bigg(1+\frac{a}{2\ell}\bigg)^2
    -\frac{x^2}{\ell^2}\bigg)}\simeq\frac{\ell^2}{a^2}\sqrt{\bigg(1-\frac{x^2}{\ell^2}\bigg)\bigg(1
          -\frac{\ka^2x^2}{\ell^2}\bigg)}\nonumber\\
  \fl
        &=:\frac{\ell^2}{a^2}\,J(x)\,,
          \label{Jdef}
\end{eqnarray}
with
\begin{equation}\label{kappa}
  \ka:=\bigg(1+\frac a{2\ell}\bigg)^{-1}=\frac{2\ell}{2\ell+a}=1-\frac1N
\end{equation}
(cf.~Fig.~\ref{fig.Jnell} left).
\begin{figure}
  \includegraphics[width=.49\textwidth]{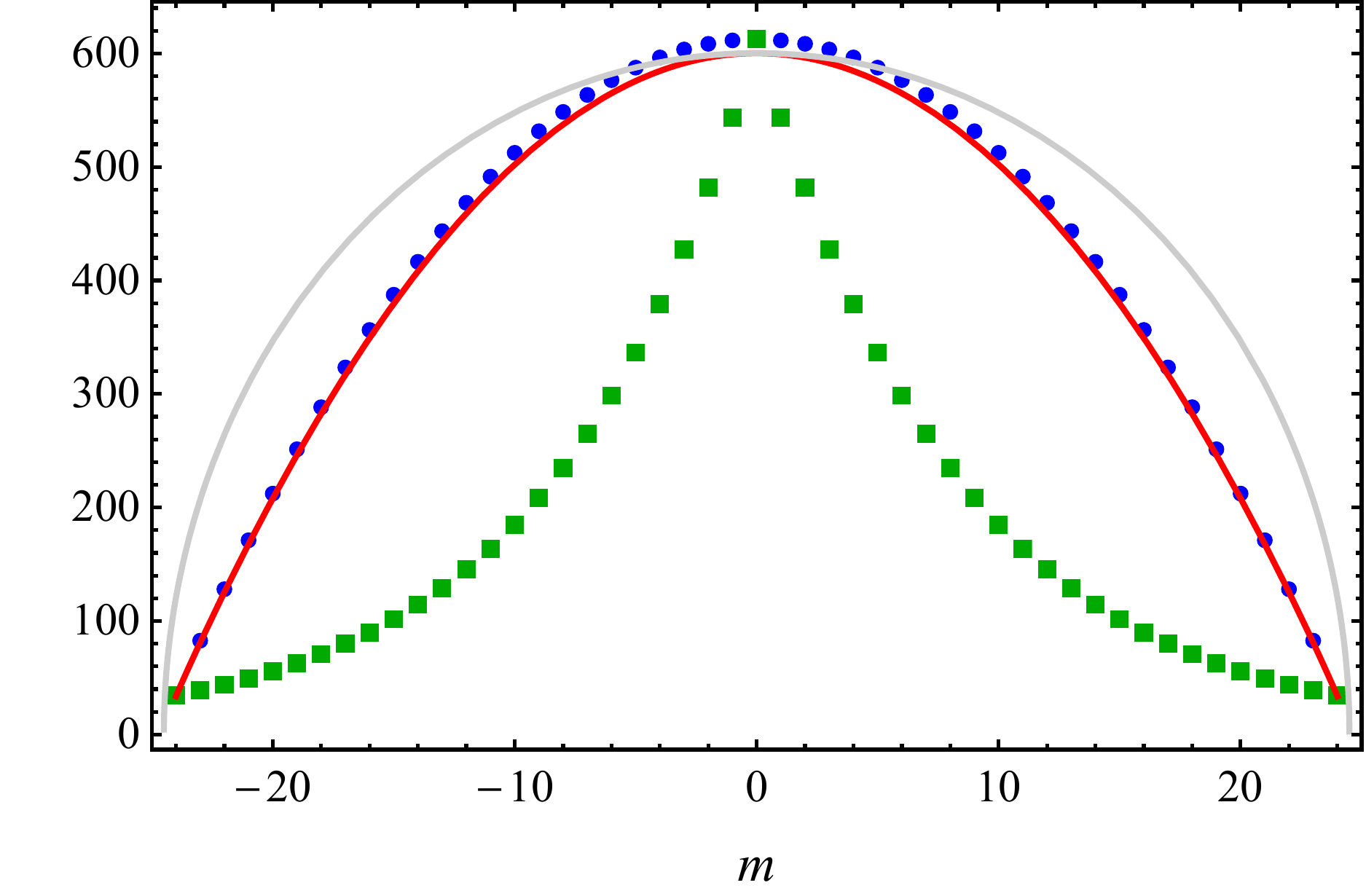}
  \hfill
  \includegraphics[width=.49\textwidth]{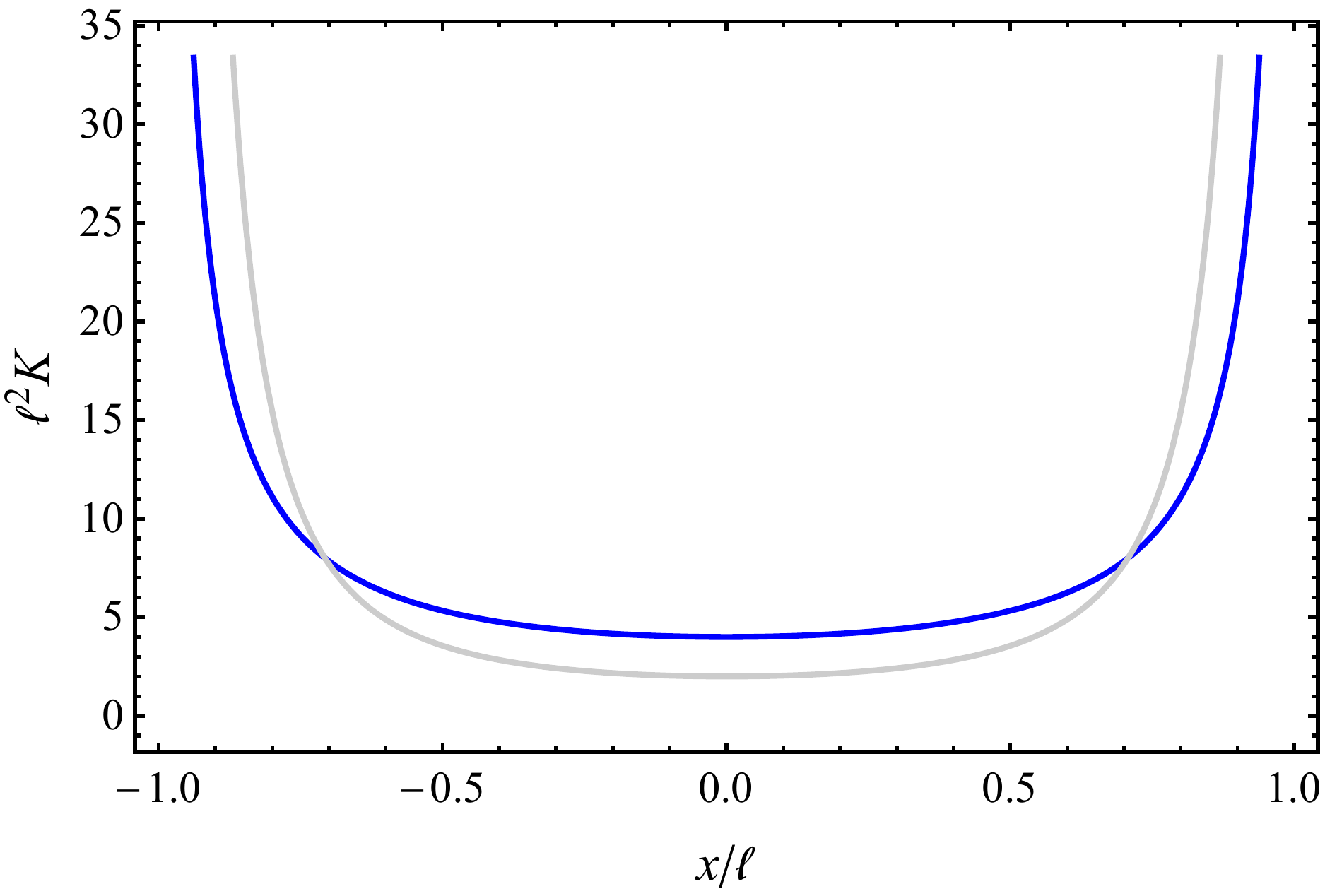}
  \caption{Left: couplings $J_{m+N/2-1}$ of the Lamé chain~\eqref{Lamechain} with $k^2=1/2$ (blue
    dots), compared to those of the rainbow chain coinciding with the former at $m=0,\pm(N/2-1)$
    (green squares), the continuum approximation $\ell^2J(m)$ in Eq.~\eqref{Jdef} for $N=50$ (red
    line), and the (scaled) Fermi velocity $\ell^2v_F(m)=\sqrt{\ell^2-m^2}$ (gray line) of a gas
    of free fermions trapped by the harmonic potential $V(x)=x^2$. Right: scalar curvature of the
    background space associated with the latter Lamé chain in the limit $N\to\infty$ (blue line)
    and with the gas of free fermions in a harmonic potential (gray line).}
\label{fig.Jnell}
\end{figure}
Note that, although~$\ka\to1$ as $N\to\infty$ and therefore $J(x)\simeq1-x^2/\ell^2$, we have not
replaced $\ka$ by $1$ in $J(x)$ since we need $\int_{0}^\ell J(x)^{-1}\diff x$ to be convergent
(see below). Note also that the limiting form of $J(x)$ is reminiscent of the (appropriately
scaled) Fermi velocity $v_F(x)=\sqrt{1-x^2/\ell^2}$ of a gas of free fermions trapped in the
harmonic potential $V(x)=x^2$~\cite{DSVC17}, the main difference being that in the latter case
$\tilde\ell=\ell\pi/2$ is finite.

Substituting Eq.~\eqref{Jdef} into~\eqref{HJxa} we thus have
\begin{equation}
  \label{Hcont}
  H\simeq\frac{\iu\ell^2}a\int_{-\ell}^\ell J(x)\big(\psiL^\dagger(x)\overset{\leftrightarrow}{\pd_x}\psiL(x)
  -\psiR^\dagger(x)\overset{\leftrightarrow}{\pd_x}\psiR(x)\big)\diff x\,,
\end{equation}
with $J(x)$ given by Eq.~\eqref{Jdef}. Using the boundary conditions~\cite{RRS15}
\[
  \psiL(\pm\ell)=\pm\iu\psiR(\pm\ell)
\]
and integrating by parts we obtain the equivalent expression
\begin{eqnarray}
  \label{Hcont2}
  H\simeq\frac{2\iu\ell^2}a\int_{-\ell}^\ell \Big[&J(x)\big(\psiL^\dagger(x)\pd_x\psiL(x)
  -\psiR^\dagger(x)\pd_x\psiR(x)\big)\\
  &+\frac{J'(x)}2\big(\psiL^\dagger(x)\psiL(x)-\psiR^\dagger(x)
  \psiR(x)\big)\big]\diff x\,,
\end{eqnarray}
where $J':=\pd_xJ$. The key observation in Ref.~\cite{RDRCS17} is that the Lagrangian density
associated with the Hamiltonian~\eqref{Hcont2}, namely (up to inessential multiplicative
constants)
\begin{equation}\label{cLJ}
  \fl
  \cL=\psiL^\dagger\pd_t\psiL+\psiR^\dagger\pd_t\psiR-J\big(\psiL^\dagger\pd_x\psiL
  -\psiR^\dagger\pd_x\psiR\big)-\frac{J'}2\big(\psiL^\dagger\psiL-\psiR^\dagger
  \psiR\big)\,,
\end{equation}
coincides with that of a free massless Dirac fermion in a curved background space with an
appropriate metric. To see this in our case, and to compute the background metric, we recall the
expression for the latter Lagrangian density:
\[
  \cL_F=e\,\overline\Psi\slashed D\Psi,\qquad\Psi:=
  \left(
    \psiL\atop\psiR
  \right),
  \quad \overline\Psi:=\Psi^\dagger\ga^0\,.
\]
Here\footnote{We are mostly following the notation of Ref.~\cite{BD82}, which slightly differs
  from that of Ref.~\cite{RDRCS17}.} $e=\det(e_\mu^a)$ is the determinant of the components of the
dual $e^a_\mu$ of the zweibein $E_a^\mu=g^{\mu\nu}\eta_{ab}e^b_\nu$ (with $a,\mu\in\{0,1\}$) and
$\slashed D=E^\mu_a\ga^aD_\mu$, where summation over repeated indices is implied. The $\ga$
matrices are $\ga^0=\iu\si^x$, $\ga^1=\si^y$, and
\[
  D_\mu=\pd_\mu+\frac18\,\om^{ab}_\mu[\ga_a,\ga_b]\,,
\]
(with $\ga_a=\eta_{ab}\ga^b$ and $(\eta_{ab})=\diag(-1,1)$), where $\om^{ab}_\mu$ is the spin
connection. The background metric $g_{\mu\nu}$ is then given by
\[
  g_{\mu\nu}=\eta_{ab}e^a_\mu e^b_\nu\,,
\]
and the spin connection is determined by the metric through the equations
\begin{equation}\label{omab}
  \om^{ab}_\mu=e^a_\nu(\nabla_\mu E^b)^{\nu}=e^a_\nu\big(\pd_\mu
  E^{b\nu}+\Ga^{\nu}_{\la\mu}E^{b\la}\big) =-\om^{ba}_\mu\,,
\end{equation}
where $(\Ga^{\nu}_{\la\mu})$ are the Christoffel symbols of the metric $g_{\mu\nu}$. If we assume
that the zweibein is such that the matrix $(E^\mu_a)$ (and hence $(e_\mu^a)$) is diagonal, the
Lagrangian density~$\cL_F$ reduces to
\[
  \cL_F=-\Psi^\dagger\bigg(e^1_1\pd_t+e^0_0\si^z\pd_x+\frac12e^1_1\om_0^{01}\si^z
  +\frac12e^0_0\om^{01}_1\bigg)\Psi\,.
\]
Comparing with Eq.~\eqref{cLJ} we arrive at the system
\begin{equation}\label{eJconds}
  e^1_1=-1,\quad e_0^0=J\,,\quad \om_0^{01}=-J',\quad \om^{01}_1=0\,.
\end{equation}
The metric is then given by
\[
  g_{00}=-(e_0^0)^2=-J^2\,,\quad g_{11}=(e_1^1)^2=1\,,\quad g_{01}=g_{10}=0\,,
\]
and hence
\begin{equation}\label{ds2}
  \diff s^2=-J^2(x)\diff t^2+\diff x^2\,.
\end{equation}
The non-vanishing Christoffel symbols are
\[
  \Ga^{0}_{01}=\Ga^{0}_{10}=\pd_x\log J\,,\qquad \Ga^1_{00}=\frac12\pd_xJ^2\,,
\]
from which it easily follows that the last two equations in~\eqref{eJconds} are consistent
with~\eqref{omab}. The Ricci tensor of the background manifold is given by
\[
  R_{00}=JJ''\,,\quad R_{11}=-\frac{J''}J\,,\quad R_{01}=R_{10}=0,
\]
and hence the scalar curvature reads
\[
  R=g^{\mu\nu}R_{\mu\nu}=-2\frac{J''}J\,.
\]
Setting~$\xi:=x/\ell$, $\ka'^2:=1-\ka^2$, from Eq.~\eqref{Jdef} it easily follows that
\[
  R=\frac2{\ell^2}\,\frac{(1-\xi^2)^2\big(2-\ka'^2-2(1-2\ka'^2)\xi^2\big)
    +\ka'^4(3-2\xi^2)\xi^4}%
  {(1-\xi^2)^2(1-\ka^2\xi^2)^2}\simeq\frac4{\ell^2-x^2}\,,
\]
so that $R>0$ everywhere and $R\to\infty$ for $x\to\pm\ell$ (cf.~Fig.~\ref{fig.Jnell} right). This
is in sharp contrast with the analogous result for the rainbow chain studied in
Refs.~\cite{RRS14,RRS15,RDRCS17,TRS18}, for which $J(x)=-\e^{-h|x|}$ and consequently
\[
  R=4h\de(x)-2h^2
\]
is negative for $x\ne0$ and singular at the origin. Again, the formula for the scalar curvature of
the model under study resembles that of the free fermion gas trapped by a harmonic potential
studied in Ref.~\cite{DSVC17}, which in appropriate units is given by $R=2\ell^2(\ell^2-x^2)^{-2}$
(cf.~Fig.~\ref{fig.Jnell} right).

To obtain an asymptotic formula for the Rényi entanglement entropy $S_{A,\al}$ of the Lamé
chain~\eqref{Lamechain} with $k^2=1/2$, we pass to the conformally flat form of the
metric~\eqref{ds2}
\[
  \diff s^2=J^2(-\diff t^2+\diff\tilde x^2)\,.
\]
through the change of variable
\[
  \tilde x:=\int_{0}^x J(y)^{-1}\diff y\,.
\]
Using again Eq.~\eqref{Jdef} we then obtain
\[
  \tilde x=\ell\arcsn(x/\ell;\ka)=\ell F(\arcsin(x/\ell);\ka)\,,
\]
where
\[
  F(\th;\ka):=\int_0^{\th}\frac{\diff\vp}{1-\ka^2\sin^2\vp}
\]
is the incomplete elliptic integral of the first kind. Hence
$\tilde x\in[-\tilde\ell,\tilde\ell\,]$, where the conformal length~$\tilde\ell$ is given by
\[
  \tilde\ell=\ell F(\pi/2;\ka)=K(\ka)\ell=K\bigl(1-N^{-1}\bigr)\ell.
\]
Note that as $N\to\infty$
\begin{equation}\label{tildex}
  \tilde x\simeq\ell\arctanh(x/\ell)\,,
\end{equation}
except near $x=\pm\ell$. On the other hand, the conformal length~$\tilde\ell$ diverges
logarithmically as $N\to\infty$; more precisely, we have~\cite{OLBC10}
\[
  K\bigl(1-N^{-1}\bigr)=\frac12\log N+O(1)\,.
\]
This behavior is again quite different from that of the rainbow chain, for which $\tilde\ell$ is
finite (in fact, independent of $N$).

Since the Lagrangian density~$\cL$ associated with the continuum limit of the Lamé chain under
study coincides with that of a massless Dirac fermion in the curved background with metric
$\diff s^2=J^2(-\diff t^2+\diff\tilde x^2)$, the $N\to\infty$ behavior of the bipartite
entanglement entropy $S_A$ of the former model can be analyzed by studying the \emph{Euclidean}
action corresponding to the Lagrangian density~$\cL_F=e\,\overline\Psi\slashed D\Psi$. The latter
action can be written in complex isothermal
coordinates 
as~\cite{DSVC17,RDRCS17}
\[
  \cS=\frac1{2\pi}\int
  J(x)\Big(\psiL^\dagger\overset{\leftrightarrow}{\pd_{z}}\psiL+
  \psiR^\dagger\overset{\leftrightarrow}{\pd_{\bar z}}\psiR\Big)\diff z\wedge\diff\bar z\,,
\]
up to inessential constant factors. According to the result in Ref.~\cite{TRS18}, the Rényi
entanglement entropy $S_{A,\al}$ of this model for a bipartition in which $A=[-l,-x]$ behaves as
\begin{equation}\label{SACFT}
  S_{A,\al}=\frac1{12}\big(1+\al^{-1}\big)\log\biggl(\frac{2\tilde\ell}{\eta\pi} J(x)\cos\biggl(\frac{\pi
    \tilde x}{2\tilde\ell}\biggr)\biggr)\,,
\end{equation}
where $\eta$ is an ultraviolet cutoff independent of $x$ and $\ell$. As explained above, from the
latter formula we can deduce an asymptotic approximation for the Rényi entanglement entropy of the
Lamé chain~\eqref{Lamechain} with $k^2=1/2$ at half filling in the limit $N\to\infty$, for a
bipartition with $A=\{0,\dots,L-1\}$. Indeed, defining $\Lr:=L/N\in(0,1)$ we have
\[
  \frac{\ell-x}{2\ell}=\Lr\en\implies\en\frac{x}{l}=1-2\Lr\,,
\]
and therefore
\[
  J(x)\simeq 1-\frac{x^2}{\ell^2}=4\Lr(1-\Lr)\,,
\]
up to $O(N^{-1})$ terms. Using Eq.~\eqref{tildex} for $\tilde x$ and setting
$\ell=a(N-1)/2\simeq aN/2$ we obtain
\begin{equation}
  \fl
  S_{A,\al}\simeq\frac1{12}\big(1+\al^{-1}\big)\log\biggl\{\frac{4}\pi\Lr(1-\Lr)NK\bigl(1-\tfrac1N\bigr)
  \cos\biggl(\frac{\pi\arctanh(1-2\Lr)}{2K\bigl(1-\tfrac1N\bigr)}\biggr)\biggr\}+\ga_\al,
\end{equation}
where $\ga_a$ is a non-universal constant independent of $x$, $\ell$ and $N$. In fact, except for
values of $\Lr$ very close to $0$ or $1$ we can discard the $\cos$ term in the logarithm, since
as $N\to\infty$ it is of order $(\log N)^{-2}$. Thus a simpler but still sufficiently accurate
asymptotic formula for $S_A$ is
\begin{equation}
  \label{SAasymp}
  S_{A,\al}\simeq\frac1{12}\big(1+\al^{-1}\big)\log\biggl[\frac{4}\pi\Lr(1-\Lr)NK
  \bigl(1-\tfrac1N\bigr)\biggr]+\ga_\al,
\end{equation}
or equivalently
\begin{equation}\label{SAas2}
  S_{A,\al}\simeq\frac1{12}\big(1+\al^{-1}\big)\log\biggl[\frac{2}\pi\Lr(1-\Lr)N\log
  N\biggr]+\tilde\ga_\al,
\end{equation}
where $\tilde\ga_\al$ is another constant independent of $x$, $\ell$ and $N$.
Thus
\[
  S_{A,\al}\simeq\frac1{12}\big(1+\al^{-1}\big)\log N+O(\log\log N),
\]
where the first term is characteristic of a critical system with central charge $c=1$, in the same
universality class as a free fermion with open boundary conditions. Note, however, that the
divergent part of $S_{A,\al}$,
\[
  S_{A,\al}\simeq\frac1{12}\big(1+\al^{-1}\big)\log(N\log N),
\]
fundamentally differs from the well-known $\log N$ behavior found in the homogeneous XX chain and
most one-dimensional critical models\footnote{An exception is, for instance, the model studied in
  Ref.~\cite{XAS18}. Note, however, that in this model the $\log(\log N)$ correction only arises
  after projecting to a sector with well-defined magnetization.}.

\begin{figure}[t]
  \includegraphics[width=.49\textwidth]{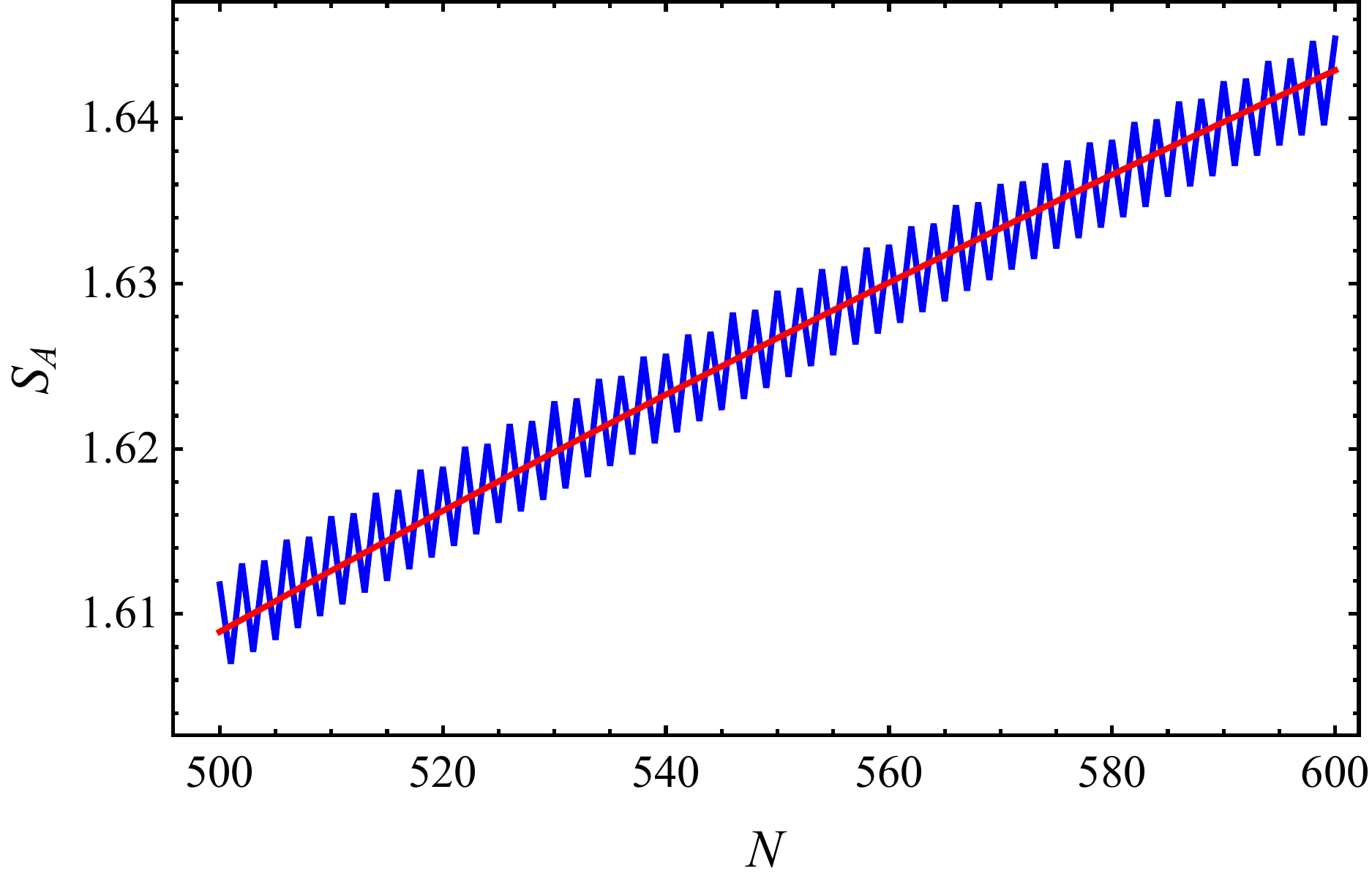}\hfill
  \includegraphics[width=.49\textwidth]{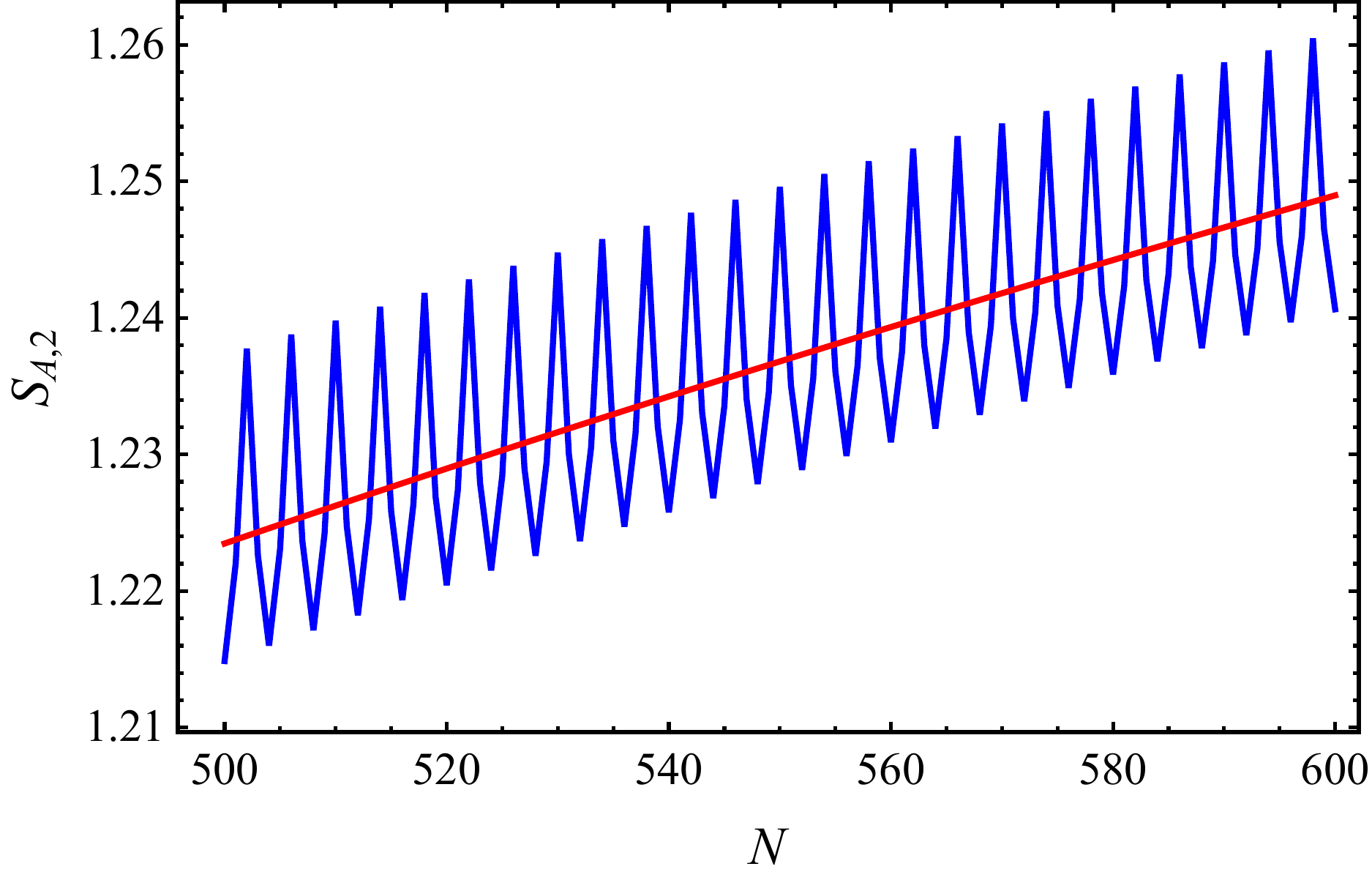}
  \caption{Rényi entanglement entropy of a single block of $\lfloor N/2\rfloor$ successive spins
    starting at the left end of the Lamé chain~\eqref{Lamechain} with $k^2=1/2$ and
    $500\le N\le 600$ spins at half filling ($M=\lfloor N/2\rfloor$), for $\al=1$ (left) and
    $\al=2$ (right), compared to its CFT-based approximation~\eqref{SAasymp} (continuous red
    line).}
  \label{fig.entent}
\end{figure}

Using Eq.~\eqref{SA}, we have numerically computed the Rényi entanglement entropy of a block of
$L$ consecutive spins at the left end of the Lamé chain~\eqref{Lamechain} with $k^2=1/2$ at half
filling for several values of the Rényi parameter $\al$ and the relative block length $\Lr=L/N$,
with $N$ up to $600$ spins. To this end, it is necessary to compute the roots of the critical
polynomial $P_N$ with very high accuracy, since in general the correlation matrix $C_A$ has a
significant number of eigenvalues very close to $0$ or $1$. More precisely, we have found it
necessary to work with $4N$ significant digits in the numerical computation of the roots of $P_N$
and the subsequent numerical diagonalization of the correlation matrix $C_A$. In general, the
agreement of the numerical values of $S_{A,\al}$ thus obtained with the CFT asymptotic
approximation~\eqref{SAasymp} is quite good, particularly for $\al\le1$. For instance, in
Fig.~\ref{fig.entent} we compare $S_{A,\al}$ with $\al=1,2$ to the latter CFT formula for
$\Lr=1/2$ and $N$ ranging from $500$ to $600$, where the non-universal parameter~$\ga_\al$ in
Eq.~\eqref{SAasymp} is estimated through a standard least squares fit of the data. It is apparent
that the fit is excellent in both cases, the coefficient of variation (i.e., the root mean squared
error divided by the mean, in percentage points) being equal to $0.164559$ for $\al=1$ and
$0.670771$ for $\al=2$. Of course, what these comparisons actually test is whether $S_{A,\al}$
behaves as
\[
  S_{A,\al}\simeq\frac1{12}\big(1+\al^{-1}\big)\log\Bigl(N
  K\bigl(1-\tfrac1N\bigr)\Bigr)+\text{const.},
\]
not the specific dependence of the constant in the right-hand side with $\Lr$ in
Eq.~\eqref{SAasymp}. To ascertain the latter dependence, it suffices to note that if
Eq.~\eqref{SAasymp} holds the value of the parameter $\ga_\al$ should not depend on $\Lr$. In view
of this observation, by way of example we have compared the value of $\ga_\al$ in the range
$0.1\le\al\le5$ (at intervals of $0.05$) obtained fitting Eq.~\eqref{SAasymp} to $S_{A,\al}$ with
$500\le N\le 600$ for $\Lr=1/4$ and $\Lr=1/2$. As is apparent from Fig.~\ref{fig.gammafit} (left),
both values differ by less than $10^{-2}$ for $\al\ge0.25$ and by about $3\times 10^{-4}$ for
$\al=5$, in excellent agreement with the $\Lr$ dependence of $S_{A,\al}$ predicted by
Eq.~\eqref{SAasymp}.
\begin{figure}[t]
  \includegraphics[width=.49\textwidth]{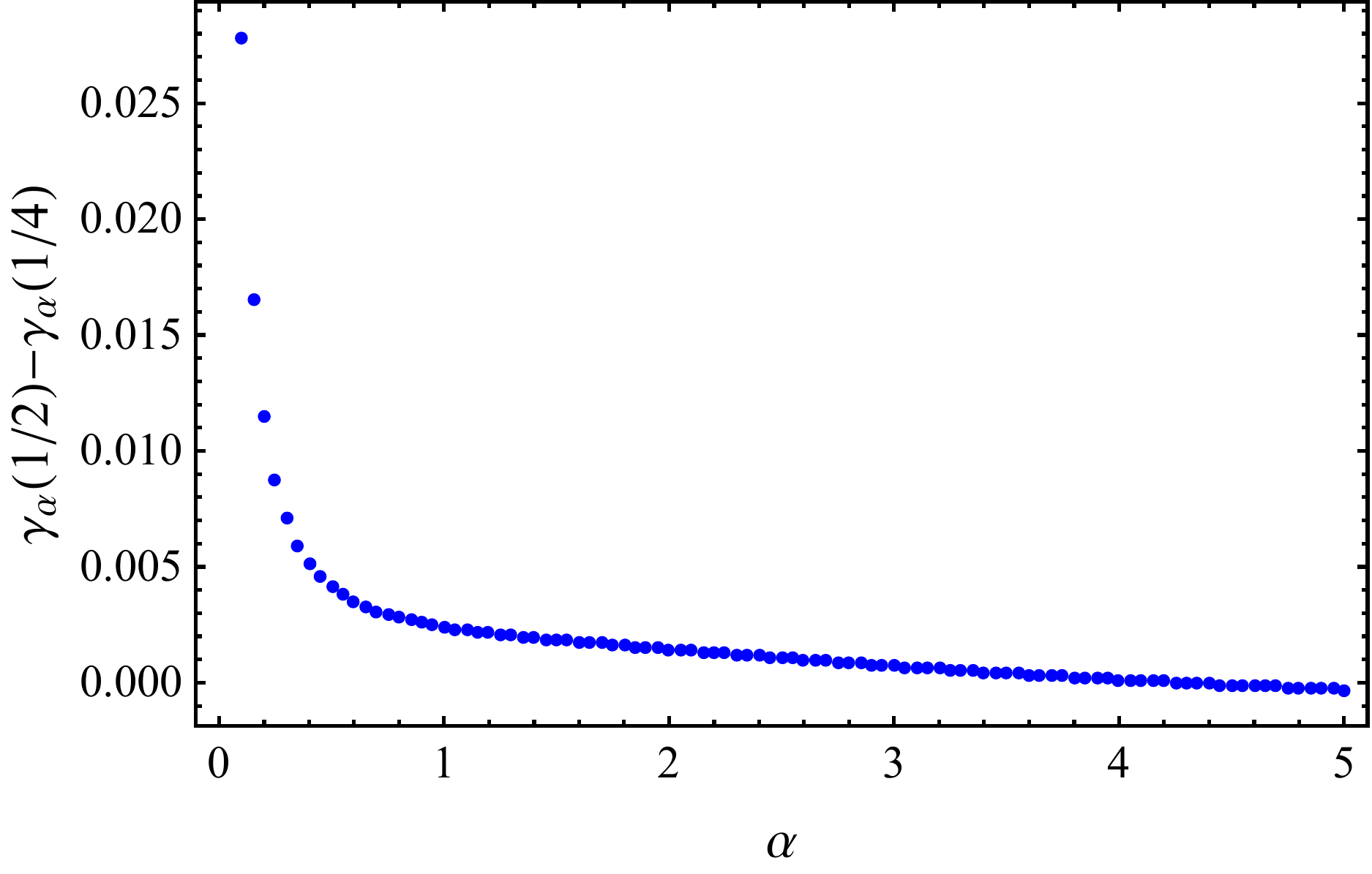}\hfill
  \includegraphics[width=.49\textwidth]{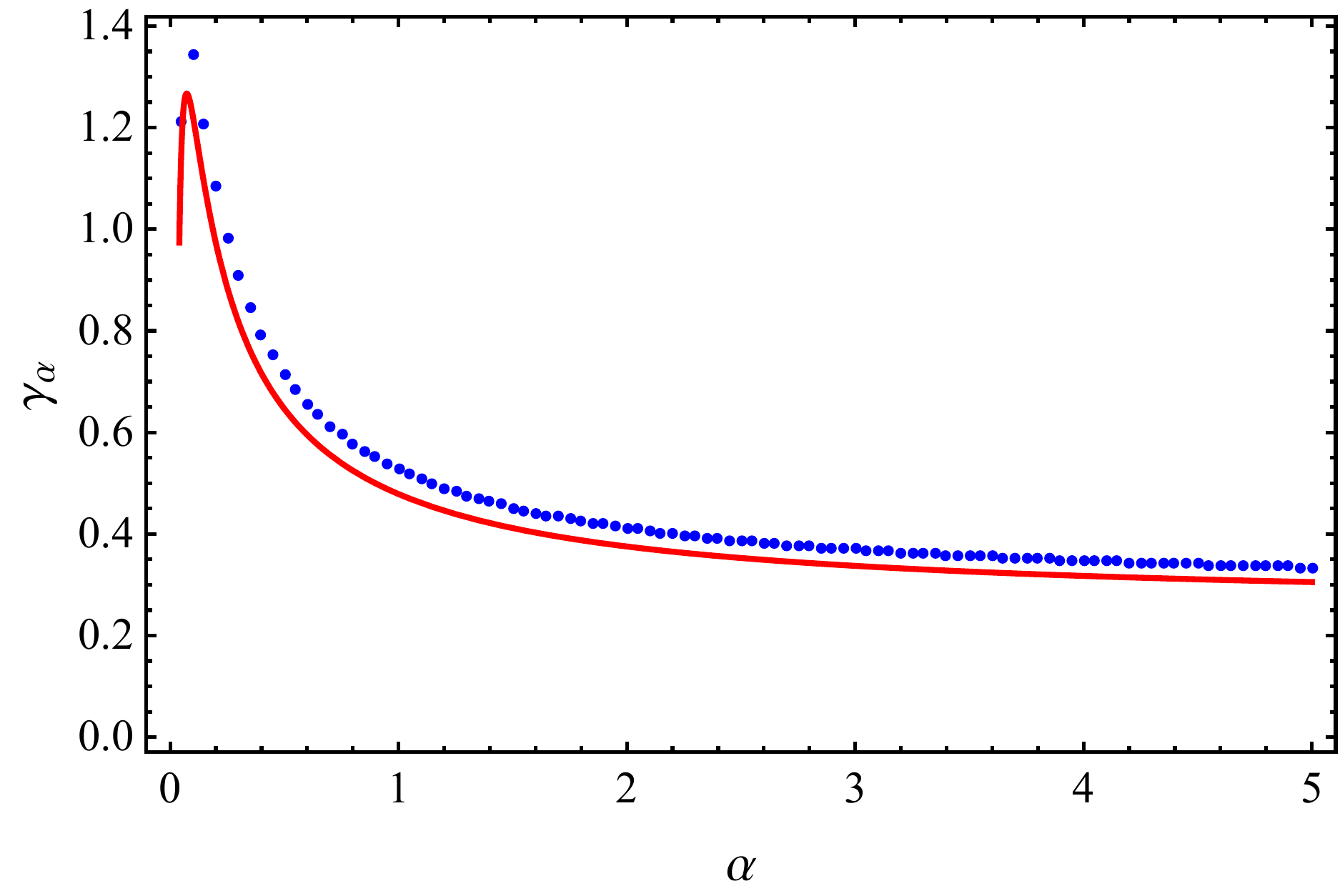}
  \caption{Left: difference between the parameter $\ga_\al$ in Eq.~\eqref{SAasymp} in the range
    $0.1\le\al\le5$ (at intervals of $0.05$) for $\Lr=1/2$ and $\Lr=1/4$. In both cases, $\ga_\al$
    was computed by fitting $S_{A,\al}$ for $500\le N\le600$ to Eq.~\eqref{SAasymp}. Right:
    parameter $\ga_\al$ for the Lamé chain~\eqref{Lamechain} with $k^2=1/2$ (blue dots) compared
    to the analogous quantity~\eqref{gammaun} for the uniform XX chain (red line).}
  \label{fig.gammafit}
\end{figure}%
Note, finally, that the numerical value of $\ga_\al$ for the Lamé chain~\eqref{Lamechain} with
$k^2=1/2$ (for, e.g., $\Lr=1/2$) is remarkably close to the exact value of its counterpart for the
uniform XX chain, namely~\cite{FC11}
\begin{eqnarray}
  \fl
  \ga_\al=\frac12\bigg(1+\frac1\al\bigg)\bigg\{
  &\frac13\log 2\nonumber\\
  \fl
  &+\int_0^\infty\bigg[\frac{\al}{1-\al^2}\,\Big(\al\csch t-\csch(t/\al)\Big)\csch
    t-\frac{\e^{-2t}}6\bigg]\,\frac{\diff t}t\bigg\}
  \label{gammaun}
\end{eqnarray}
(cf.~Fig.~\ref{fig.gammafit}, right).

\section{Conclusions and outlook}\label{sec.conc}

In this work we have established a connection between inhomogeneous XX spin chains (or free
fermion systems) and quasi-exactly solvable models on the line constructed from the $\sla(2)$
algebra. Indeed, any such model generating a family of weakly orthogonal polynomials defines a
corresponding XX chain, whose single-particle Hamiltonian is determined by the coefficients of the
three-term recursion relation of the polynomial family. Moreover, two realizations of the same QES
model equivalent under a projective transformation give rise to isomorphic chains. We have
classified all QES models on the line giving rise to a weakly orthogonal polynomial system under
projective transformations, finding six inequivalent families. Each of them generates a
corresponding family of inhomogeneous XX chains, whose hopping amplitudes and on-site energies are
simple algebraic functions of the chain sites. Although in some cases the hopping amplitudes of
these chains coincide with those of the chains constructed from the classical Krawtchouk and dual
Hahn polynomials in Ref.~\cite{CNV19}, their on-site energies differ. Thus the six types of XX
chains introduced in this paper appear to be new. In particular, from these six new types one can
construct two families of XX chains associated with different QES realizations of the well-known
Lamé (finite gap) potential on the line.

From the polynomial family associated with an inhomogeneous XX chain it is straightforward to
construct the correlation matrix of the corresponding free fermion system, whose eigenvalues yield
its entanglement spectrum. In fact, this is one of the most efficient methods for computing the
bipartite Rényi entanglement entropy of such models. We have used this method to analyze the
entanglement entropy of one of the new Lamé chains, whose on-site energies vanish for a suitable
value of the modulus of the elliptic function. This makes it possible to apply the CFT techniques
in Ref.~\cite{TRS18} to find an asymptotic formula for the entanglement entropy at half filling
when the number of sites $N$ tends to infinity, which reproduces with great accuracy the numerical
results. Interestingly, we show that although the leading behavior of the entropy is the
characteristic one for a critical one-dimensional model with $c=1$, there is a correction
proportional to $\log(\log N)$ which is unusual for this type of systems.

The above results suggest several possible lines for future research. To begin with, the CFT
techniques applied in this work to approximate the entanglement entropy of one of the Lamé chains
can also be used for the new chain associated with the well-known sextic QES potential, whose
coefficients depend on a free parameter. In particular, it could be of interest to ascertain if in
this case there is also a subleading $\log(\log N)$ correction to the leading $\log N$ behavior of
the entanglement entropy. It would also be natural to explore whether the above field-theoretic
techniques can be generalized to chains with non-vanishing on-site energies, and to arbitrary
Fermi momentum. Of course, the bipartite entanglement entropy is only the simplest type of
multipartite entropy one can consider, and in fact the asymptotic behavior of the multi-block
Rényi entanglement entropies of the homogeneous XX model and similar free fermion systems have
been widely studied~(see, e.g., \cite{CH09,ATC10,FC10,CE10,AEF14,CFGT17}). A similar analysis for
the Lamé chain introduced in this paper, or the new chain constructed from the sextic QES
potential, would therefore be worth pursuing. Finally, another natural problem to investigate is
whether any of the chains introduced in this work allows for perfect state transfer of spin
excitations~\cite{Bo07,Ka10,CJ10,Je11}. Indeed, it is known that a necessary condition for this to
happen is that both the hopping amplitude and on-site energy be symmetric about the center of the
chain~\cite{VZ12}. This is actually the case for many of the models introduced in this paper (for
suitable values of the parameters), including the two families of Lamé chains.

\ack This work was partially supported by Spain's Mi\-nis\-te\-rio de Ciencia, Innovaci\'on y
Universidades under grant PGC2018-094898-B-I00, as well as by Universidad Complutense de Madrid
under grant G/6400100/3000. The authors would like to thank L. Vinet for pointing out to us
Ref.~\cite{CNV19} and for inspiring conversations. They are also grateful to G. Sierra for his
helpful comments and for bringing Ref.~\cite{XAS18} to our attention.

\appendix

\section{Classification of QES models admitting a weakly orthogonal polynomial family}

In this appendix we provide the details of the classification in Section~\ref{sec.class} of QES
models on the line giving rise to a weakly orthogonal polynomial system
(cf.~Table~\ref{tab.canP}). As explained in the latter section, these models are characterized by
the fact that the quartic polynomial $P(z)$ in Eq.~\eqref{HgPQR} vanishes at zero and
infinity\footnote{Recall that in this context one says that the polynomial $P(z)$ vanishes at
  $z=\infty$ if $\tP(w):=w^4 P(1/w)$ vanishes at $w=0$, i.e., if $\deg P<4$. The order of
  $z=\infty$ as a root of $P(z)$, defined as the order of $w=0$ as a root of $\tP(w)$, is equal to
  $4-\deg P$.}. Moreover, two such models are equivalent if their polynomials $P(z)$ and $\tP(w)$
are related by a real projective transformation~\eqref{tPQR}. We thus need to find all equivalence
classes of real polynomials $P(z)$ of degree at most four vanishing at zero and infinity (and such
that $P(z)$ is positive in some open interval), modulo real projective
transformations~\eqref{tPQR}. In fact, the classification in Section~\ref{sec.class} follows
easily by considering the root pattern of $P$ in the extended real line, which is invariant under
projective transformations. Let us encode such a pattern by a list of positive integers
$(m_1,m_2,\dots,m_r)$, where $r\ge2$ is the number of distinct real roots of $P$ and $m_i$ is the
multiplicity of the $i$-th root. From the previous remarks it follows that the only allowed root
patterns are
\[
  (2,1,1),\quad (2,2),\quad (3,1),\quad (1,1,1,1),\quad (1,1)\,.
\]
We shall next see that the first root pattern gives rise to the first two canonical forms in
Table~\ref{tab.canP}, while each of the remaining patterns respectively yields the canonical forms
$3$ to $6$. It shall be convenient to deal separately with the cases in which I)~$P$ has at least
one multiple real root, and II)~all real roots of $P$ are simple.

\medskip
\noindent I) $P$ has (at least) one multiple real root

This case corresponds to the first three root patterns above. Applying if necessary a projective
transformation of the form $w=(z-a)^{-1}$, we can assume that $\infty$ is a multiple root of $P$,
or equivalently that $\deg P\le2$. If $\deg P=1$ (corresponding to the root pattern $(3,1)$) then
$P(z)=\nu z$ with $\nu\ne0$, which is in turn mapped to $\tP(w)=w$ (case 4 in
Table~\ref{tab.canP}) by the dilation $z=\nu w$. If $\deg P=2$, apart from the root at the origin
$P$ must have an additional finite real root at $z=-a$, so that $P(z)=cz(z+a)$ with $c\ne0$. The
dilation $z=\la w$ then maps $P(z)$ to
\[
\tP(w)=\frac{c}{\la^2}\,\la w(\la w+a)=c w\bigg(w+\frac a\la\bigg)\,.
\]
If $a=0$ we obtain the third canonical form in Table~\ref{tab.canP} (note that in this case
$c=\nu$ must be positive, because otherwise $P$ would be nonpositive everywhere). If $a\ne0$,
setting $\la=\sgn c\cdot a$ we have $\tP(w)=|c|w(1+\sgn c\cdot w)$, which yields the first
canonical form for $c>0$ and the second one for $c<0$. This exhausts case I, since $\deg P=0$ if
and only if $P=0$.

\medskip
\noindent I) $P$ has no multiple real roots

There are two subcases to consider, depending on whether $P$ has four simple real roots or two
real and two complex conjugate roots (including the root at infinity). In the first case (which
corresponds to the root pattern $(1,1,1,1)$), up to a dilation we can write
\[
  P(z)=c z(z+1)(z+a)\,,
\]
with $c\ne0$ and $a\ne0,1$. To begin with, we can assume that $c>0$, since the linear map
$z=-w-1$ transforms $P(z)$ into $\tP(w)=-c w(w+1)(w+1-a)$. Let us show, finally, that we can
take $a\in(0,1)$. Indeed, if $a>0$ we apply the inversion $z=1/w$, which maps $P(z)$ into
\[
  \tP(w)=c w^3\bigg(\frac1w+1\bigg)\bigg(\frac1w+a\bigg)=c a w(w+1)\bigg(w+\frac1a\bigg).
\]
Since $c a>0$, and $1/a\in(0,1)$ if $a\notin(0,1)$, we see that $\tilde P$ coincides with the
fifth canonical form in this case (with $a$ replaced by $1/a$). Finally, if $a<0$ we perform the
projective transformation $z=-a(w+1)/w$, under which $P(z)$ is mapped to
\[
  \fl \tP(w)=\frac{c}{a^2}\,w^4(-a)\bigg(\frac1w+1\bigg)\bigg(1-a-\frac aw\bigg)\bigg(-\frac
  aw\bigg) =c(1-a)w(w+1)\bigg(w-\frac a{1-a}\bigg).
\]
Again, since $a<0$ we have $c(1-a)>0$ and
\[
  -\frac{a}{1-a}=\frac{|a|}{1+|a|}\in(0,1)\,,
\]
so that $\tP$ adopts the fifth canonical form in Table~\ref{tab.canP}.

Consider, finally, the case in which $P$ has two complex conjugate and two real roots (necessarily
at $0$ and $\infty$), corresponding to the last root pattern $(1,1)$. We can thus write
\[
  P(z)=c z(z^2+2az+b^2)\,,
\]
with $c\ne0$ and $b>|a|$. We can obviously assume that $c>0$ (otherwise, apply the
transformation $z=-w$). The dilation $z=b w$ then maps $P(z)$ into
\[
  \tP(w)=\frac c{b^2}\,b w(b^2w^2+2ab w+b^2)=c b
  w\bigg(w^2+\frac{2a}b\,w+1\bigg)\,,
\]
with $c b>0$ and $|a|/b<1$, which coincides with the sixth canonical form in
Table~\ref{tab.canP}.

\section*{References}


\begin{thebibliography}{10}
\providecommand{\url}[1]{\texttt{#1}}
\providecommand{\urlprefix}{URL }
\providecommand{\eprint}[2][]{\url{#2}}

\bibitem{CC04JSTAT}
Calabrese P and Cardy J, \emph{Entanglement entropy and quantum field theory,}
  2004 \emph{J. Stat. Mech.-Theory E.} \textbf{2004} P06002(27)

\bibitem{CC09}
Calabrese P and Cardy J, \emph{Entanglement entropy and conformal field
  theory,}  2009 \emph{J. Phys. A: Math. Theor.} \textbf{42} 504005(36)

\bibitem{JK04}
Jin B~Q and Korepin V~E, \emph{Quantum spin chain, {T}oeplitz determinants and
  the {F}isher--{H}artwig conjecture,}  2004 \emph{J. Stat. Phys.} \textbf{116}
  79

\bibitem{CE10}
Calabrese P and Essler F~H~L, \emph{Universal corrections to scaling for block
  entanglement in spin-$1/2$ {$XX$} chains,}  2010 \emph{J. Stat. Mech.-Theory
  E.} \textbf{2010} P08029(28)

\bibitem{FC11}
Fagotti M and Calabrese P, \emph{Universal parity effects in the entanglement
  entropy of {$XX$} chains with open boundary conditions,}  2011 \emph{J. Stat.
  Mech.-Theory E.} \textbf{2011} P01017(26)

\bibitem{CH09}
Casini H and Huerta M, \emph{Remarks on the entanglement entropy for
  disconnected regions,}  2009 \emph{J. High Energy Phys.} \textbf{2009}
  048(18)

\bibitem{CCT09}
Calabrese P, Cardy J and Tonni E, \emph{Entanglement entropy of two disjoint
  intervals in conformal field theory,}  2009 \emph{J. Stat. Mech.-Theory E.}
  \textbf{2009} P11001(38)

\bibitem{ATC10}
Alba V, Tagliacozzo L and Calabrese P, \emph{Entanglement entropy of two
  disjoint blocks in critical {I}sing models,}  2010 \emph{Phys. Rev. B}
  \textbf{81} 060411(R)(4)

\bibitem{FC10}
Fagotti M and Calabrese P, \emph{Entanglement entropy of two disjoint blocks in
  {$XY$} chains,}  2010 \emph{J. Stat. Mech.-Theory E.} \textbf{2010}
  P04016(35)

\bibitem{AEF14}
Ares F, Esteve J~G and Falceto F, \emph{Entanglement of several blocks in
  fermionic chains,}  2014 \emph{Phys. Rev. A} \textbf{90} 062321(8)

\bibitem{CFGT17}
Carrasco J~A, Finkel F, Gonz{\'a}lez-L{\'o}pez A and Tempesta P, \emph{A
  duality principle for the multi-block entanglement entropy of free fermion
  systems,}  2017 \emph{Sci. Rep.-UK} \textbf{7} 11206(11)

\bibitem{FH68}
Fisher M~E and Hartwig R~E, \emph{Toeplitz determinants: some applications,
  theorems and conjectures,}  1968 \emph{Adv. Chem. Phys.} \textbf{15} 333

\bibitem{Ba79}
Basor E~L, \emph{A localization theorem for {T}oeplitz determinants,}  1979
  \emph{Indiana Math. J.} \textbf{28} 975

\bibitem{DIK11}
Deift P, Its A and Krasovsky I, \emph{Asymptotics of {T}oeplitz, {H}ankel, and
  {T}oeplitz$\,+${H}ankel determinants with {F}isher--{H}artwig singularities,}
   2011 \emph{Ann. Math.} \textbf{174} 1243

\bibitem{VRL10}
Vitagliano G, Riera A and Latorre J~I, \emph{Volume-law scaling for the
  entanglement entropy in spin-$1/2$ chains,}  2010 \emph{New J. Phys.}
  \textbf{12} 113049(16)

\bibitem{RDRCS17}
Rodr{\'{i}}guez-Laguna J, Dubail J, Ram{\'{i}}rez G, Calabrese P and Sierra G,
  \emph{More on the rainbow chain: entanglement, space-time geometry and
  thermal states,}  2017 \emph{J. Phys. A: Math. Theor.} \textbf{50} 164001(18)

\bibitem{DSVC17}
Dubail J, St{\'{e}}phan J~M, Viti J and Calabrese P, \emph{Conformal field
  theory for inhomogeneous one-dimensional quantum systems: the example of
  non-interacting {F}ermi gases,}  2017 \emph{SciPost Phys.} \textbf{2} 002(21)

\bibitem{TRS18}
Tonni E, Rodr{\'{i}}guez-Laguna J and Sierra G, \emph{Entanglement hamiltonian
  and entanglement contour in inhomogeneous 1{D} critical systems,}  2018
  \emph{J. Stat. Mech.-Theory E.} \textbf{2018} 043105(39)

\bibitem{CJ10}
Chakrabarti R and {Van der Jeugt} J, \emph{Quantum communication through a spin
  chain with interaction determined by a {J}acobi matrix,}  2010 \emph{J. Phys.
  A: Math. Theor.} \textbf{43} 085302(20)

\bibitem{Je11}
Van~der Jeugt J, \emph{Quantum communication and state transfer in spin
  chains,}  2011 \emph{J. Phys. Conf. Ser.} \textbf{284} 012059(10)

\bibitem{VZ12}
Vinet L and Zhedanov A, \emph{How to construct spin chains with perfect state
  transfer,}  2012 \emph{Phys. Rev. A} \textbf{85} 012323(7)

\bibitem{CNV19}
Cramp{\'{e}} N, Nepomechie R~I and Vinet L, \emph{Free-fermion entanglement and
  orthogonal polynomials,}  2019 \emph{J. Stat. Mech.-Theory E.} \textbf{2019}
  093101(17)

\bibitem{Tu88}
Turbiner A~V, \emph{{Q}uasi-exactly solvable problems and $\mathfrak{sl}(2)$
  algebra,}  1988 \emph{Commun. Math. Phys.} \textbf{118} 467

\bibitem{Sh89}
Shifman M~A, \emph{{N}ew findings in quantum mechanics (partial algebraization
  of the spectral problem),}  1989 \emph{Int. J. Mod. Phys. A} \textbf{4} 2897

\bibitem{ST89}
Shifman M~A and Turbiner A~V, \emph{{Q}uantal problems with partial
  algebraization of the spectrum,}  1989 \emph{Commun. Math. Phys.}
  \textbf{126} 347

\bibitem{Us94}
Ushveridze A~G, \emph{{Q}uasi-{E}xactly {S}olvable {M}odels in {Q}uantum
  {M}echanics} (Bristol: Institute of Physics Publishing) 1994

\bibitem{BD96}
Bender C~M and Dunne G~V, \emph{Quasi-exactly solvable systems and orthogonal
  polynomials,}  1996 \emph{J. Math. Phys.} \textbf{37} 6

\bibitem{FGR96}
Finkel F, Gonz{\'a}lez-L{\'o}pez A and Rodr{\'\i}guez M~A, \emph{Quasi-exactly
  solvable potentials on the line and orthogonal polynomials,}  1996 \emph{J.
  Math. Phys.} \textbf{37} 3954

\bibitem{Ar64}
Arscott F~M, \emph{Periodic Differential Equations} (Oxford: Pergamon) 1964

\bibitem{VLRK03}
Vidal G, Latorre J~I, Rico E and Kitaev A, \emph{Entanglement in quantum
  critical phenomena,}  2003 \emph{Phys. Rev. Lett.} \textbf{90} 227902(4)

\bibitem{Pe03}
Peschel I, \emph{Calculation of reduced density matrices from correlation
  functions,}  2003 \emph{J. Phys. A: Math. Gen} \textbf{36} L205

\bibitem{Ch78}
Chihara T~S, \emph{{A}n {I}ntroduction to {O}rthogonal {P}olynomials} (New
  York: Gordon and Breach) 1978

\bibitem{GKO93}
Gonz{\'a}lez-L\'opez A, Kamran N and Olver P~J, \emph{{N}ormalizability of
  one-dimensional quasi-exactly solvable {S}chr{\"o}dinger operators,}  1993
  \emph{Commun. Math. Phys.} \textbf{153} 117

\bibitem{GKO94} Gonz{\'a}lez-L\'opez A, Kamran N and Olver P~J, \emph{{Q}uasi-exact solvability,}
  1994 \emph{Contemp. Math.} \textbf{160} 113

\bibitem{Kr29}
Krawtchouk M, \emph{Sur une g{\'e}n{\'e}ralisation des polyn{\^o}mes
  d’{H}ermite,}  1929 \emph{C. R. Hebd. Seances Acad. Sci.} \textbf{189} 620

\bibitem{ha49}
Hahn, \emph{{\"U}ber {O}rthogonalpolynome, die $q$‐{D}ifferenzengleichungen
  gen{\"u}gen,}  1949 \emph{Math. Nachr.} \textbf{2} 4

\bibitem{AGI83}
Alhassid Y, G{\"{u}}rsey F and Iachello F, \emph{Potential scattering, transfer
  matrix, and group theory,}  1983 \emph{Phys. Rev. Lett.} \textbf{50} 873

\bibitem{BB93}
Braibant S and Brihaye Y, \emph{Quasi-exactly-solvable system and sphaleron
  stability,}  1993 \emph{J. Math. Phys.} \textbf{34} 2107

\bibitem{GKLS97}
Greene P, Kofman L, Linde A and Starobinsky A, \emph{Structure of resonance in
  preheating after inflation,}  1997 \emph{Phys. Rev. D} \textbf{56} 6175

\bibitem{FGLR00}
Finkel F, Gonz{\'{a}}lez-L{\'{o}}pez A, Maroto A~L and Rodr{\'{i}}guez M~A,
  \emph{The {L}am{\'{e}} equation in parametric resonance after inflation,}
  2000 \emph{Phys. Rev. D} \textbf{62} 103515(7)

\bibitem{FGR00}
Finkel F, Gonz{\'a}lez-L{\'o}pez A and Rodr{\'\i}guez M~A, \emph{A new
  algebraization of the {L}am{\'e} equation,}  2000 \emph{J. Phys. A: Math.
  Gen.} \textbf{33} 1519

\bibitem{NC10}
Nielsen M~A and Chuang I~L, \emph{Quantum Computation and Quantum Information}
  (Cambridge: Cambridge University Press), 10th {A}nniversary edition 2010

\bibitem{Kat12}
Katsura H, \emph{Sine-square deformation of solvable spin chains and conformal
  field theories,}  2012 \emph{J. Phys. A: Math. Theor.} \textbf{45} 115003(17)

\bibitem{RRS14}
Ram{\'{i}}rez G, Rodr{\'{i}}guez-Laguna J and Sierra G, \emph{From conformal to
  volume law for the entanglement entropy in exponentially deformed critical
  spin $1/2$ chains,}  2014 \emph{J. Stat. Mech.-Theory E.} \textbf{2014}
  P10004(15)

\bibitem{RRS15}
Ram{\'{i}}rez G, Rodr{\'{i}}guez-Laguna J and Sierra G, \emph{Entanglement over
  the rainbow,}  2015 \emph{J. Stat. Mech.-Theory E.} \textbf{2015} 06002(20)

\bibitem{BD82}
Birrell N~D and Davies P~C~W, \emph{Quantum Fields in Curved Space} (Cambridge:
  Cambridge University Press) 1982

\bibitem{OLBC10}
Olver F~W~J, Lozier D~W, Boisvert R~F and Clark C~W, eds., \emph{{NIST}
  Handbook of Mathematical Functions} (Cambridge University Press) 2010

\bibitem{XAS18}
Xavier J~C, Alcaraz F~C and Sierra G, \emph{Equipartition of the entanglement
  entropy,}  2018 \emph{Phys. Rev. B} \textbf{98} 041106(R)(6)

\bibitem{Bo07}
Bose S, \emph{Quantum communication through spin chain dynamics: an
  introductory overview,}  2007 \emph{Contemp. Phys.} \textbf{48} 13

\bibitem{Ka10}
Kay A, \emph{Perfect, efficient, state transfer and its application as a
  constructive tool,}  2010 \emph{Int. J. Quantum Inf.} \textbf{8} 641

\end{thebibliography}

\end{document}